%% file: main.tex
\documentclass[conference]{IEEEtran}

\IEEEoverridecommandlockouts
\usepackage{cite}
\usepackage{amsmath,amssymb,amsfonts}
\usepackage{algorithmic}
\usepackage{graphicx}
\usepackage{textcomp}
\usepackage{xcolor}
\def\BibTeX{{\rm B\kern-.05em{\sc i\kern-.025em b}\kern-.08em
    T\kern-.1667em\lower.7ex\hbox{E}\kern-.125emX}}

\input{mydefs}

\begin{document}
\bstctlcite{IEEEexample:BSTcontrol}

\title{\method: Locally Differentially Private Graph Synthesis 
by Exploiting Neighborhood Structure
}

\author{
\IEEEauthorblockN{
Jiawei Dong\IEEEauthorrefmark{1},
Zhikun Zhang\IEEEauthorrefmark{1},
Quan Yuan\IEEEauthorrefmark{1}\IEEEauthorrefmark{4},
Zhe Liu\IEEEauthorrefmark{1},
Yunjun Gao\IEEEauthorrefmark{1}\IEEEauthorrefmark{2}
}

\IEEEauthorblockA{
\IEEEauthorrefmark{1}\textit{Zhejiang University},
\IEEEauthorrefmark{4}\textit{Hong Kong Baptist University}, %
\textit{\{jwdong, zhikun, yq21, zhe.liu, gaoyj\}@zju.edu.cn}
}
\thanks{\IEEEauthorrefmark{2}Corresponding author.}
}

\maketitle

\begin{abstract}
The widespread application of graph data inevitably brings significant privacy risks, as its unprotected use can lead to the leakage of sensitive information. 
These risks are particularly acute in the setting with an untrusted curator, where the data remains decentralized and each user only holds the connections to their neighbors.
To mitigate such privacy risks, we adopt \textit{local differential privacy} (LDP) to collect users' private information and generate a synthetic graph.
However, existing methods suffer from either excessive noise injection by perturbing the local adjacency lists or significant structural information loss due to the simplistic graph encoding process.
To address these issues, we propose \method, an effective graph synthesis algorithm that takes one step further by exploiting neighborhood structures under LDP.
To obtain neighborhood statistics beyond degrees, \method aggregates a noisy global view from perturbed adjacency lists and combines it with projected local connections to estimate node-level triangle counts.
To correct the structural inconsistency caused by separately perturbing degree and triangle count, \method jointly refines them into feasible degree-triangle targets.
To reconstruct a global graph from these estimated targets, \method adopts a triangle-first strategy that first preserves local clustering structures and then fulfills remaining degree requirements.
Extensive experiments on four real-world datasets and multiple commonly used graph metrics validate the superiority of \method.
Source code is available at https://github.com/ZJU-TrustAID/LDPGraph.
\end{abstract}


\section{Introduction}
\label{sec:introduction}

As a data structure specifically designed to describe and manage relationships between entities, graph data is playing a pivotal role in numerous real-world systems, such as social networks~\cite{leskovec2010signed, jiang2016big, sekara2016fundamental}, 
advertising recommendations~\cite{leskovec2007dynamics, yin2023next},
bioinformatics~\cite{Barabasi2004Network,zitnik2018modeling}, \etc
In the real world, a large amount of graph data is not centrally available; instead, they are distributed across individual users who hold only a local view of their immediate connections, such as personal mobile contact lists, face-to-face interactions, or peer-to-peer communications.
Directly collecting this sensitive data for centralized analysis is infeasible due to a lack of user trust and significant privacy concerns~\cite{zhu2023blink,sun2019analyzing}.
This reality challenges traditional privacy-preserving techniques (\eg, anonymization~\cite{liu2008towards,cheng2010k,wang2014high} or standard differential privacy~\cite{dwork2006calibrating,dwork2014algorithmic}) which typically operate under the strong and often impractical assumption of a trusted central curator that has aggregated the complete graph.

To address this issue, \textit{local differential privacy} (LDP)~\cite{Kasiviswanathan2011What, Duchi2013Local} has been adopted as the de facto mechanism for the collection and analysis of private graph data.
Under the LDP setting, each user perturbs their own data directly on their local device before transmission, ensuring that sensitive raw information never leaves their control.
Existing studies on LDP-based graph analysis mainly focus on designing tailored algorithms for specific analysis tasks,
such as triangle counting~\cite{imola2021locally,imola2022communication,wei2025privacy}, 
subgraph counting~\cite{sun2019analyzing},
community detection~\cite{hou2023wdt, fu2024community}, \etc
Differentially private data synthesis provides a general paradigm for supporting diverse downstream analyses and has been explored in a range of prior studies~\cite{zhang2021privsyn, wang2023privtrace, du2023LDPTrace, Hu2024Real,yuan2025privload,Ye2026PrivAGS}.
Our paper extends this idea to graph data under LDP by collecting statistics from users and generating a synthetic graph.
The newly generated graph preserves the essential properties of the original graph with LDP guarantee. 
This paradigm offers a more flexible and scalable alternative to existing approaches by enabling a wide range of downstream analyses from a privacy-preserving output.

With a limited local view, it is challenging to obtain effective topological information from the global graph that is essential to construct a synthetic graph.
Without external information, each user 
cannot access
the graph structure beyond their local neighborhood.
This naturally yields two fundamental types of property each user can report: the degree representing the number of its neighbors and the edge set identifying its neighbors.
These two basic information can inspire two intuitive strategies for graph synthesis under LDP.

\mypara{Existing Solutions}
The first is \textit{degree-based generation} (\dbg), where
each user perturbs and reports only their degree.
A curator then feeds these noisy degrees into a graph generation model,
resulting in a synthetic graph that approximates the original degree distribution.
While \dbg is efficient in terms of communication cost and privacy budget, this method suffers from severe information loss, since it reduces the connectivity of each user to a single scalar.
Another strategy is \textit{adjacency list-based generation} (\albg), where each user perturbs and reports its adjacency list.
Intuitively, by capturing the most fine-grained information, \albg could achieve a perfect reconstruction of the original graph.
However, the high dimensionality of adjacency lists introduces excessive noise under a practical privacy budget.

To address the dilemma of ``too little information'' or ``too much noise,'' Qin \etal~\cite{qin2017ldpgen} proposed \gen by operating on intermediate-level information (\ie, clustering).
This intermediate representation captures more structural details than \dbg, while maintaining a lower information dimension compared to \albg.
However, the accuracy of their clustering operation under perturbation is typically low, and the reliance on simple degree sequences still discards too much of the original graph's essential structural information. 

\mypara{Our Proposal}
The limitations of previous studies motivate us to rethink what structural information can be effectively captured from the rigid perspective of LDP.
As established in the LDP setting for graphs, each user's knowledge is confined to their immediate one-hop connections.
The degree representation is too simplistic, while the high-dimensional adjacency list is fragile to the noise.
We observe that the key lies in looking just one step further: beyond the edges connecting a user to their neighbors, we need the connection information between the neighbors themselves.  
We call this richer, more detailed local topology the \textit{neighborhood structure}. 
This structure can be powerfully summarized by two fundamental statistics: the degree representing the neighborhood's size and the triangle count representing its internal density.
Capturing both statistics is critical for reconstructing graphs with realistic neighborhood structures.
However, this introduces significant challenges rooted in the LDP model.

Firstly, neighborhood structure is difficult to estimate under LDP.
Each user only has direct access to their own one-hop view, while the connectivity among their neighbors depends on edges between other users.
Such information is not locally available and must be inferred from privatized reports.
This makes accurate neighborhood structure estimation challenging.
Furthermore, the estimated neighborhood structures may become internally inconsistent after perturbation.
It is characterized by degree and triangle count, which are inherently coupled.
However, LDP noise is injected into these statistics separately.
Consequently, they may no longer correspond to any feasible local structure.
Therefore, these statistics must be jointly refined before they can serve as reconstruction targets.
The final challenge lies in reconstructing a complete graph from the refined neighborhood structures.
Even if refined targets are accurate, they do not specify the detailed edge set.
The reconstruction algorithm must satisfy node-level degree requirements and high-order triangle constraints simultaneously, resulting in a globally valid graph.

To address these issues, we propose \method, consisting of a carefully designed two-round information collection protocol and a subsequent graph reconstruction algorithm.
To privately obtain the neighborhood structure information, 
which is impossible with local one-hop data alone, \method collects a noisy global view in the first round.
In the second round, \method addresses the high sensitivity of triangle counting by employing a projection mechanism. 
In particular, a threshold optimization mechanism is designed to effectively reduce the noise impact while mitigating the error introduced by information loss.
Users then combine this noisy global view with their projected local connections to estimate triangle counts.
Furthermore, to ensure the statistical feasibility of the collected data, \method corrects the inconsistency between the noisy degrees and triangle counts.
Based on these privatized and refined statistics, \method then employs a hierarchical reconstruction strategy to synthesize the entire graph. 
The primary objective of this reconstruction phase is to recover the neighborhood structure by matching the degrees and triangle counts.
By building these fundamental local structures, the synthesized graph can effectively reproduce crucial topological properties of the original graph, such as transitivity and modularity.

\mypara{Evaluation}
To validate the effectiveness of \method, we conduct experiments on four real-world graph datasets and multiple commonly used metrics. 
The experimental results show that \method outperforms the state-of-the-art methods in most settings.
For instance, when the privacy budget is 1.0, 
\method achieves 65\% lower relative error of modularity and 80\% lower relative error of transitivity than those of the best baseline on the Facebook dataset.
Then, we conduct an ablation study on the key components of \method. 
Specifically, we investigate the impact of the threshold optimization and consistency restoration on improving the quality of collected statistics.
The results validate the effectiveness of each careful design and its crucial role in the superior performance.

\mypara{Contributions}
In summary, the main contributions of this paper are three-fold: 
\begin{itemize}
[itemsep=2pt,topsep=2pt,parsep=0pt]
    \item We advance the conceptual basis of LDP graph synthesis by taking one step further: from preserving one-hop statistics to preserving the rich two-hop neighborhood structure.
    
    \item We propose a practical method \method that utilizes a two-round protocol to collect neighborhood statistics 
    under LDP while ensuring their quality with projection threshold optimization and structural consistency restoration.
    Based on this, \method introduces a hierarchical reconstruction that not only faithfully preserves these collected local neighborhood structures, but also effectively reproduces topological properties of the original graph.
    
    \item We conduct extensive experiments %
    on multiple datasets and metrics to illustrate the effectiveness of our proposed \method.
\end{itemize}

\section{Preliminaries}
\label{sec:preliminary}

\subsection{Local Differential Privacy}
\label{subsec:local differential privacy}
The concept of \textit{differential privacy} (DP)~\cite{dwork2006calibrating} has emerged as the gold standard for privacy-preserving data analysis. 
In the traditional centralized setting, a trusted data curator holds the entire original dataset and injects carefully calibrated noise into the query results before releasing them. 

\begin{definition}[$\varepsilon$-Differential Privacy]
{An algorithm $\mathcal{M}$ satisfies $\varepsilon$-differential privacy, where $\varepsilon>0$, if and only if for any two neighboring datasets $D$ and $D^{\prime}$, the following holds:} 
\begin{equation*}
    \forall T \subseteq Range (\mathcal{M}):\Pr {\mathcal{M}(D) \in T }
    \le e^{\varepsilon} \Pr {\mathcal{M}(D^{\prime}) \in T},
\end{equation*}
where Range $(\mathcal{M})$ denotes the set of all possible outputs of the algorithm $\mathcal{M}$, and $D$, $D^{\prime}$ satisfy $D = D^{\prime} + r $ or $ D^{\prime} = D + r$.
\end{definition}

However, the trusted curator assumption does not hold in most real-world scenarios. 
\textit{Local differential privacy} (LDP)~\cite{Kasiviswanathan2011What} addresses this by shifting the responsibility for data perturbation from the server to the individual user. 
Under the LDP setting, each user applies a randomization mechanism to their own data before sending it to the untrusted curator. 
This ensures that the sensitive raw data never leaves the user's device, providing a much stronger privacy guarantee.
Formally, the definition of LDP is as follows:
\begin{definition}[$\varepsilon$-Local Differential Privacy]
A mechanism $\mathcal{M}$ satisfies $\varepsilon$-LDP if for any two input $v_1, v_2$ from the input domain, and for any possible output $y$, the following holds:
\begin{equation*}
    \forall y \subseteq Range (\mathcal{M}):\Pr {\mathcal{M}(v_1) = y} \leq e^{\varepsilon} \Pr {\mathcal{M}(v_2) = y}.
\end{equation*}
\end{definition}
Essentially, LDP ensures an untrusted curator observing output $y$ cannot confidently infer if the user's true input was $v_1$ or $v_2$, thus protecting privacy.
Several fundamental mechanisms can serve as building blocks for complex LDP algorithms.

\mypara{Randomized Response}
\textit{Randomized response} (RR)~\cite{warner1965randomized} is a representative mechanism to perturb \textit{categorical data}. 
For a user with a binary value $v \in \{0, 1\}$, RR works as follows:
\begin{equation*}
\label{eq:rr}
\Pr{\mathcal{M}_{RR}(v) = v'} = 
\begin{cases} 
    p, & 	\text{if}~ v'= 1-v \\
    1-p, & 	\text{if}~ v' = v 
\end{cases}
\end{equation*}

By setting $p = \frac{e^\varepsilon}{e^\varepsilon + 1}$, this mechanism satisfies $\varepsilon$-LDP.

\mypara{Laplace Mechanism}
\textit{Laplace mechanism} (LM)~\cite{dwork2014algorithmic} 
satisfies the LDP requirements by adding random Laplace noise to the local \textit{numerical value} for each user. 
The sanitized value is $\mathcal{M}_f(D) = f(D) + Lap\left(\frac{\triangle_f}{\varepsilon}\right)$, 
where $Lap(\beta)$ stands for a random variable sampled from the Laplace distribution $\Pr {Lap(\beta)=x}=\frac{1}{2\beta}e^{-\left | x \right | / \beta}$, 
and $\triangle_f$ stands for the maximum possible change in the results of a query across any two neighbor databases
(\ie sensitivity). 

\mypara{Composition Properties of LDP}
The following composition properties of LDP are commonly used for building complex LDP algorithms from simpler subroutines.

\begin{itemize}[itemsep=0pt,topsep=2pt,parsep=0pt]
    \item \mypara{Sequential Composition}
    Combining multiple subroutines that satisfy local differential privacy for
    $\{\varepsilon_{1}, \cdots,\varepsilon_{k}\}$ results in a mechanism satisfying
    $\varepsilon$-LDP for $\varepsilon = { \sum_{i}\varepsilon_{i}} $.
       
    \item \mypara{Post-processing}
    Given an $\varepsilon$-LDP algorithm $\mathcal{A}$ and any $g$, releasing $g(\mathcal{A}(D))$ still satisfies $\varepsilon$-LDP, \ie, post-processing the output of a locally differentially private algorithm does not incur additional loss of privacy. 
\end{itemize}

\subsection{LDP Graph Analysis}
\label{subsec:ldp_graph}
In real-world graph 
scenarios, such as social, collaboration, or financial networks, the sensitive information
typically resides in
the relationships between 
two individuals.
For instance, users of social platforms often wish to conceal their connections with others.
In the centralized setting, edge DP~\cite{blocki2012johnson} is utilized to protect the presence of each connection between two users.
For a local setting where each user knows only their neighbors, we have the following formal definitions:
The local information of a node $u_i \in V$ is their \textit{adjacency list},
which can be represented as an $n$-dimensional binary vector $l_i=(b_1,\dots,b_n)$, where $b_j=1$ if an edge $(u_i, u_j)\in E$ exists, and $b_j=0$ otherwise. 
Note that by convention, $b_i$ can be considered 0. 
The definition of edge LDP is as follows:
\begin{definition}[$\varepsilon$-edge LDP]
\label{def:edge_ldp}
A mechanism $\mathcal{M}$ satisfies $\varepsilon$-edge LDP if for any node, for any two of their possible adjacency lists $l$ and $l^{\prime}$ that differ in only one bit, and for any possible output $s \in Range(\mathcal{M})$, the following holds:
$$
\Pr {\mathcal{M}(l) = s} \le e^\varepsilon \cdot \Pr {\mathcal{M}(l^{\prime}) = s}.
$$
\end{definition}

\section{Problem Definition and Existing Solutions}
\label{sec:problem_definition}

\subsection{Problem Definition}
\label{subsec:problem_definition}

The core research question of this paper is to design a mechanism that allows an untrusted curator to obtain a high-fidelity synthetic graph while protecting user privacy.
This problem can be decomposed into two main objectives:
\begin{enumerate}
    \item[(i)] A curator must be able to collect information from each node $u_i$ in a manner that satisfies edge-LDP.
    \item[(ii)] Based on the perturbed data, the curator must be able to reconstruct a synthetic graph $G'=(V,E')$.
\end{enumerate}
As a result, we obtain a synthetic graph that can be used for downstream analysis tasks without any additional loss of privacy due to the post-processing property of LDP. 
We summarize the frequently used mathematical notations in~\autoref{table:math_notations}.

The utility of the synthetic graph $G'$ can be measured from multiple perspectives, such as its ability to preserve the degree distribution, the transitivity, or the community structure of $G$. 
We will introduce the evaluation metrics in detail in~\autoref{sec:evaluation}.

\begin{table}[!t]
    \centering
    \caption{Summary of mathematical notations.}
    \label{table:math_notations}
    \footnotesize
    \setlength{\tabcolsep}{1.2em}
	\begin{tabular}{cc}
		\toprule
		\textbf{Notation} & \textbf{Description}  \\
		\midrule
		$G$ & Graph \\
		$\varepsilon$ & Privacy budget  \\
		$N$ & The number of users (nodes) \\
		$d$ & The degree sequence \\
		$l$ & The adjacency list of a user (node)  \\
		$t$ & The triangle count sequence  \\
		\bottomrule
	\end{tabular}
\end{table}

\subsection{Existing Solutions}
\label{section:existing_solutions}

\mypara{DBG}
A natural idea for graph synthesis is to collect each node's degree, referred to as \textit{degree-based generation} (\dbg).
In \dbg, each user perturbs its degree under LDP with Laplace noise, and sends the result to the curator. 
The curator then aggregates all the noisy degrees to form a sequence. 
This sequence is fed into a synthetic graph generation algorithm such as BTER~\cite{seshadhri2012community}, which results in a synthetic graph that conforms to the degree distribution of the original graph.

Apparently, \dbg suffers from severe information loss, since it reduces the connectivity information of each user to a single scalar value.
While the synthetic graph might match the overall degree distribution, it inevitably fails to preserve the local clustering or global topological structure.

\mypara{ALBG} 
An alternative is to directly collect the graph's detailed edge set, which we call \textit{adjacency list-based generation} (\albg).
Each user perturbs its adjacency list using randomized response.
To handle undirected edges efficiently and avoid double reporting, an asymmetric protocol can be employed, which means each user perturbs and transmits only half of its adjacency list.
The curator then reconstructs a noisy adjacency matrix from these reports.

Since the adjacency list includes all structural information, a sufficiently large privacy budget could allow a nearly perfect synthetic graph. 
However, such large budgets are often impractical. 
Under a normal privacy budget, the protocol would introduce excessive noise.

\mypara{LDPGen~\cite{qin2017ldpgen}} 
To strike a trade-off between information loss as in \dbg and excessive noise as in \albg, Qin \etal propose \gen.
The general idea is to learn a coarse-grained structural representation through an iterative refinement process. 
In each iteration, \gen clusters users into structurally similar groups and collects noisy node-to-group connectivity with the current partitions.
Once a stable partition that preserves the graph's structure properties is obtained, a graph generation model is applied to construct the synthetic graph. 

However, \gen still relies mainly on low-dimensional connectivity patterns and degree information.
Such information is insufficient to capture richer structures.
Furthermore, as discussed in the introduction, all the above approaches only exploit the one-hop information which cannot obtain sufficient structural information, thus struggling to generate graphs with realistic neighborhood structure.

\begin{figure*}[!t]
\centering
\includegraphics[width=0.95\textwidth] {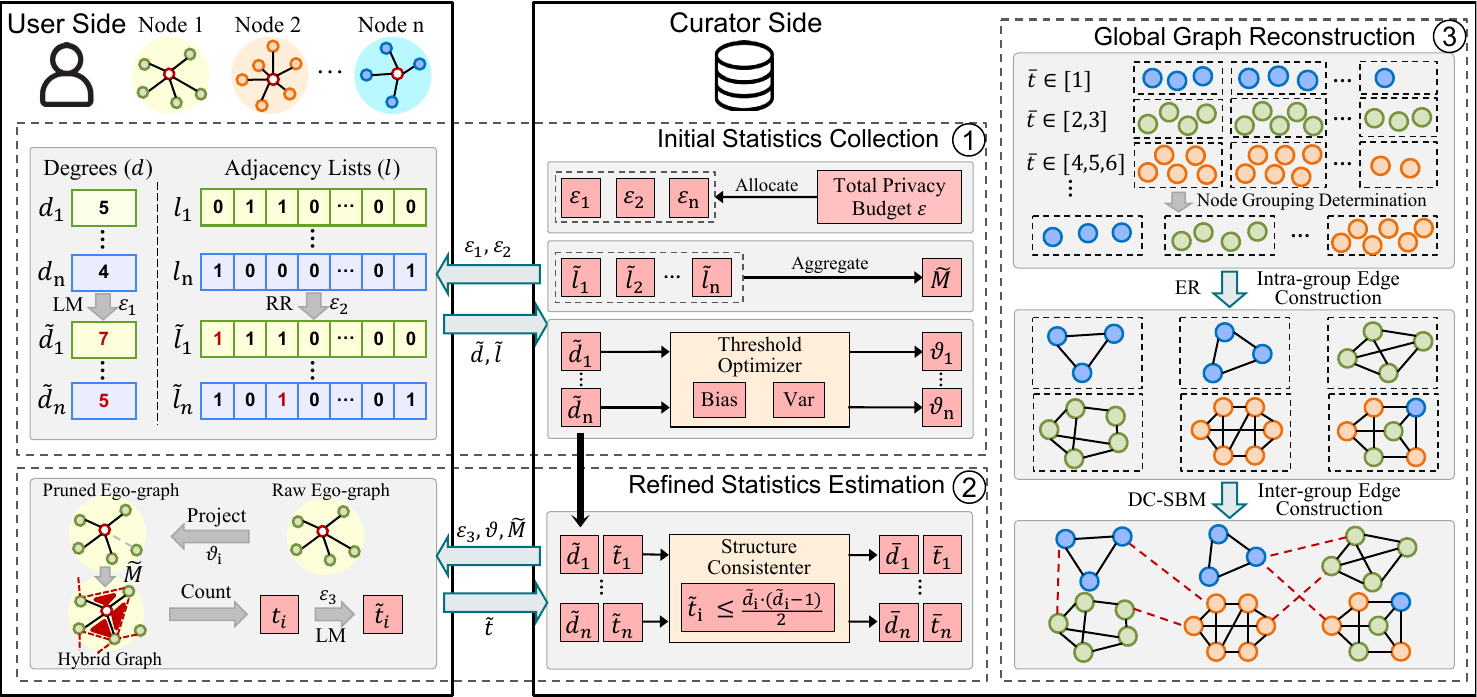}
\caption{\method overview. 
\method consists of three phases: Initial Statistics Collection, Refined Statistics Estimation, and Global Graph Reconstruction.
In Phase 1, each user reports perturbed local degrees and adjacency lists.
The curator then builds a noisy adjacency matrix and computes a projection threshold for each user.
In Phase 2, each user estimates perturbed triangle counts using the noisy adjacency matrix and projected local connections.
The curator then jointly refines noisy degrees and triangle counts into structurally consistent statistics.
In Phase 3, \method partitions nodes into groups based on refined triangle counts from Phase 2.
\method then builds dense intra-group edges to satisfy local triangle constraints, 
and concludes by adding inter-group edges to fulfill the remaining degree requirements.
LM, RR, ER, and DC-SBM stand for Laplace Mechanism, Randomized Response, Erdős–Rényi Model, and Degree-Corrected Stochastic Block Model, respectively.
}

\label{fig:framework}
\end{figure*}

\section{Our Proposal: \method}
\label{sec:methodology}

\subsection{Motivation}
Graph synthesis under LDP requires a balance between accuracy and privacy.
While generating a new graph using only basic statistics like degrees
requires less noise, this results in major information loss.
On the other hand, collecting more complex information, such as community partitions, can certainly preserve the topological characteristics of the original graph in the centralized DP setting. 
Nevertheless, this is impractical under LDP since each node is only aware of their own local connections to their neighbors.
Protecting such high-dimensional structure information demands an excessive amount of noise, which would destroy the utility of the data. 

Consequently, this trade-off motivates the need for an intermediate level of granularity that can effectively represent local structures without incurring overwhelming noise.
As established, prior approaches are fundamentally limited by their reliance on one-hop statistics, such as degrees, adjacency lists, or node-to-group connectivity.
However, realistic graph structures are not influenced solely by how many neighbors each node has, but also by how these neighbors are connected with each other.
Therefore, we take one step further and identify the local neighborhood structure as a suitable feature to target. 
Here, the local neighborhood structure of a node mainly consists of two key components: its degree and the number of triangles in which it participates. 
This allows us to break down the complex task of capturing local structure into more manageable subproblems.
By preserving these essential statistics and applying a specific reconstruction algorithm, \method can faithfully synthesize a new graph with outstanding utility while rigorously satisfying LDP.

\subsection{Overview}
\label{section:overview}

\autoref{fig:framework} illustrates the workflow of \method, which consists of three phases: initial statistics collection, refined statistics estimation, and global graph reconstruction. 
The total privacy budget $\varepsilon$ is divided into three parts: $\varepsilon_1$ for node degree collection, $\varepsilon_2$ for adjacency list perturbation, and $\varepsilon_3$ for refined statistics estimation.

\mypara{Phase 1: Initial Statistics Collection} 
The first challenge is that each user only has a local view of their immediate neighbors, making it difficult to collect neighborhood information under LDP.
To address this, each user perturbs their degree and adjacency list using LDP mechanisms and reports the results to the curator. 
The curator aggregates these into a complete noisy adjacency matrix and computes projection thresholds.
This step provides a global view and bounds the sensitivity of the estimation in next phase, ensuring that later refined statistics can be privately and accurately obtained. 

\mypara{Phase 2: Refined Statistics Estimation}
This phase aims to derive refined statistics as feasible targets for graph reconstruction.
Guided by the threshold from Phase 1, each user 
combines their projected ego-graphs with the noisy global view to estimate node-level triangle counts, perturb these counts for privacy, and report them to the curator.
Since degrees and triangle counts are perturbed separately, they may no longer characterize a valid local structure.
The curator then jointly corrects the noisy degrees and triangle counts to produce consistent and feasible statistics that capture neighborhood structures under LDP.

\mypara{Phase 3: Global Graph Reconstruction}
This phase constructs the synthetic graph from the refined statistics.
The main challenge is that degree targets and triangle targets impose different constraints on the graph.
Satisfying degrees first may destroy local clustering when triangles are later adjusted.
Therefore,
\method first partitions all nodes into groups by triangle counts.
We proceed to build intra-group edges, with connection probabilities calibrated to satisfy the target triangle constraints. 
Inter-group edges are then added to fulfill the remaining degree requirements.

\subsection{Initial Statistics Collection} 
\label{section:basic_statistics_collection}
This phase aims to collect basic statistics for later neighborhood estimation.
\method collects adjacency lists to build a noisy global view, which helps each user infer connections among its neighbors in the next phase.
Meanwhile, \method uses collected degrees to compute projection thresholds for bounding the sensitivity of triangle estimation.

\mypara{Statistics Collection}
\autoref{algorithm:basic_statistics_collection} (\autoref{line:basis_collection_begin} - \autoref{line:basis_collection_end}) shows the procedure for collecting degrees and adjacency lists.
For degree collection, each user $u_i$
perturbs their actual degree $d_i$ by adding Laplace noise with privacy budget $\varepsilon_1$.
This step can obtain a noisy degree $\tilde{d}_i$, which is reported to the curator.
For the adjacency list, we employ the randomized response mechanism to ensure LDP.
Specifically, each user $u_i$ perturbs the $j$-th bit of their adjacency list representing the edge $(i, j)$ based on probability $p=\frac{1}{1+e^{\varepsilon_{2}}}$, which is introduced in~\autoref{subsec:local differential privacy}.
The curator then aggregates all of these perturbed adjacency lists $\tilde{l}$ to construct a global noisy adjacency matrix $\tilde{M}$.

\mypara{Threshold Optimization} 
After collecting all the noisy degrees $\tilde{d}$, we try to determine a projection threshold $\theta$ for each user $u_i$.
We need to clarify the necessity of this step before we actually use the threshold in the next stage.
In the second phase, users are required to privately estimate the number of triangles in their local ego-graph.
The sensitivity of the estimation is bounded by the user's degree $d_i$,
To satisfy LDP, each user must add noise calibrated to a valid upper bound of sensitivity. 

A basic approach is to use the theoretical maximum value of degree $N-1$ as a universal bound for all users. 
This introduces substantial noise, resulting in utility degradation, especially for low-degree users.
One might consider using the noisy degree $\tilde{d}$ to reduce the bound.
Since Laplace noise could be negative, resulting in $\tilde{d}_i < d_i$, directly substituting $\tilde{d}_i$ as the upper bound leads to a violation of the LDP guarantee.

Consequently, we need a suitable strategy that projects the raw ego-graph of each user into a pruned view, ensuring the current degree is not more than a threshold $\theta$.
Then the sensitivity of the triangle count computation is upper-bounded by the threshold $\theta$.
The choice of the projection threshold $\theta$ reflects a classic trade-off between the projection error and noise error.
On the one hand, a low threshold causes more aggressive pruning of the graph.
This can easily lead to serious information loss, causing an underestimation of the true triangle count.
On the other hand, to ensure the LDP guarantee, the amount of noise added is proportional to the upper bound of sensitivity (\ie, $\theta$).
A large $\theta$ could preserve more information, but would result in excessive noise, increasing the variance of the final estimate.

Therefore, we aim to find a suitable $\theta_i$ for each user $u_i$ that minimizes the combined impact caused by projection error and noise error.
We parameterize the user-independent threshold as $\theta_i=\tilde{d}_i+\tau_i$, where $\tau_i$ is a tunable correction factor. 
The problem thus becomes finding the $\tau_i$ to minimize the function: 
\begin{equation}
f(\tau_i) = \frac{(\tilde{d}_i - \frac{1}{2})^2}{\varepsilon_1^2} e^{-\tau_i\varepsilon_1} + \frac{2}{\varepsilon_2^2} \left( \tilde{d}_i^2 + \frac{2}{\varepsilon_1^2} + 2\tilde{d}_i\tau_i + \tau_i^2 \right).
\label{eq:threshold_optimization}
\end{equation}

Finally, we obtain the suitable thresholds for each user by function \textbf{ThresholdOptimizer}, as shown in~\autoref{algorithm:basic_statistics_collection} (\autoref{line:threshold_optimization_begin} - \autoref{line:threshold_optimization_end}).
The detailed derivation will be explained in Appendix~\ref{section:theoretical_analysis}.

\begin{algorithm}[!t]
    \caption{Initial Statistics Collection}
    \label{algorithm:basic_statistics_collection}
    \KwIn{Users' adjacency lists $\{l_1,l_2,\cdots,l_n\}$, privacy budget $\varepsilon_1, \varepsilon_2$ }
    \KwOut{Noisy degrees $\tilde{d}$, perturbed adjacency matrix $\tilde{M}$, projection threshold $\theta$}
    \textbf{Step 1: Statistics Collection} \\
    \label{line:basis_collection_begin}
    \For{$i=1$ to $N$}{
        $\tilde{d}_i\leftarrow \sum_{j=1}^{n}l_{i,j} + Lap(\frac{1}{\varepsilon_1})$\\
        $\tilde{l}_i\leftarrow (RR_{\varepsilon_2}(l_{i,1}),RR_{\varepsilon_2}(l_{i,2}),\cdots,RR_{\varepsilon_2}(l_{i,i-1}))$\\
    }
    $\tilde{M}\leftarrow (\tilde{l}_1,\tilde{l}_2,\cdots, \tilde{l}_n)$\\
    \label{line:basis_collection_end}
    
    \textbf{Step 2: Threshold Optimization} \\
    \label{line:threshold_optimization_begin}
    \For{$i=1$ to $N$}{
        // Find a suitable $\tau_i$ based on \autoref{eq:threshold_optimization} \\
        $\tau_i \leftarrow $ 
        \textbf{ThresholdOptimizer}$(\varepsilon_1,\varepsilon_2,\tilde{d}_i)$\\
        $\theta_i \leftarrow \tilde{d}_i+\tau_i$ 
    }
    \label{line:threshold_optimization_end}
\end{algorithm}

\subsection{Refined Statistics Estimation} 
\label{section:refined_statistics_estimation}

After the initial statistics collection, the curator obtains noisy degrees and a noisy global adjacency view. 
In this phase, we aim to derive a set of refined statistics that can serve as feasible targets for construction.
Note that both
degree and triangle count are essential for reconstructing neighborhood structures. 
Specifically, degree captures the neighborhood size, while triangle count captures the internal connectivity among neighbors.
However, prior LDP triangle counting works \cite{imola2021locally,imola2022communication} mainly focus on the total number of triangles.
They do not guarantee that the resulting node-level statistics are structurally consistent with the degrees,
which is detrimental to graph reconstruction.
Therefore,
we refine the projection in triangle count collection and further propose the  consistency restoration to ensure the resulting statistics (\ie, triangle count and degree) can support a valid graph structure.

\mypara{Triangle Count Collection}
Each user $u_i$ first projects their local ego-graph $l_i$ according to the threshold $\theta_i$ received from the curator (denoted by function \textbf{GraphProjection} in~\autoref{line:graph_projection} of~\autoref{algorithm:triangle_counting_estimation}).
Specifically, if a user's degree $d_i$ exceeds $\theta_i$, they randomly sample and retain only $\theta_i$ neighbors;
otherwise, all neighbors are kept.
Recalling the previous phase in~\autoref{section:basic_statistics_collection}, we emphasized the necessity of the projection since this step can effectively reduce the upper bound of the sensitivity in the triangle computation. 

Subsequently, each user constructs a hybrid view by integrating two distinct sources of information:
their own actual, but pruned edges connecting to the neighbors and the noisy connections between these neighbors, which are derived from the noisy adjacency matrix $\tilde{M}$.
Within this hybrid view, all users can calculate a raw triangle count $t_{i}^{\prime}$ and the number of $2$-stars $s_i$ centered at the user $u_i$.
However, this count is biased due to the noise introduced by the Randomized Response on the edges between neighbors.
Specifically, the count includes two parts: true triangles when the connecting edges between neighbors are correctly preserved (with probability $1-p$) and false triangles formed when a pair of unconnected neighbors is reported as being connected (with probability $p$).
The expected value of the observed count is:
\begin{equation}
E[t_i^{\prime}]=(1-p)t_i+p\cdot (s_i-t_i). 
\end{equation}

Therefore, we obtain a calibrated estimate of %
triangle count:
\begin{equation}
t_i=\cfrac{t_{i}^{\prime}-p\cdot s_i}{1-2p}.
\label{eq:triangle_calibration}
\end{equation}

Since the calculation is fundamentally based on the user's authentic local connections to their neighbors, the user must perturb the raw estimate before its calibration by adding Laplace noise with a scale of $\frac{\theta_i}{\varepsilon_3}$.
The calibrated result of the private triangle count $\tilde{t}_i$ is transmitted to the curator.

\mypara{Consistency Restoration}
Upon collecting the noisy degree $\tilde{d}$ and the noisy triangle count $\tilde{t}$ from all users, a critical data consistency issue must be addressed.
For any undirected graph, the triangle count of any node is restricted by its degree through the structural relation $t\leq \frac{d(d-1)}{2}$.
However, the collected pair $(\tilde{d}_i,\tilde{t}_i)$ of user $u_i$ may violate this constraint since noise is added independently to the degree and triangle count estimates. 

To rectify this inconsistency, we perform a consistency restoration for each user, denoted by the function \textbf{Consistenter} as shown in~\autoref{algorithm:triangle_counting_estimation} (\autoref{line:consistentor_begin} - \autoref{line:consistentor_end}).
The procedure is formulated as a constrained optimization problem, seeking to find a new pair of $(\bar{d}_i, \bar{t}_i)$ that minimizes the weighted $\ell_2$ error from the original noisy collection $(\tilde{d}_i, \tilde{t}_i)$
over the feasible domain $0\leq d < N$ and $0\leq t\leq d(d-1)/2$.
The objective function is 
\begin{equation}
(\bar{d}_i,\bar{t}_i)=\mathop{\arg\min}\limits_{d,t} \left(w_d (d-\tilde{d}_i)^2+w_t(t- \tilde{t}_i)^2 \right),
\label{eq:consistency_restoration}
\end{equation}
where the weights $w_d$ and $w_t$ are set to be inversely proportional to the scale of their respective noise distributions, \ie, $\frac{w_d}{w_t}=\frac{\varepsilon_1\theta_i}{\varepsilon_3}$.
By solving this optimization problem, we obtain a set of statistics that not only respects the structural constraint of a graph but also effectively preserves the collected information, thus providing a solid foundation for the subsequent global graph reconstruction.

\begin{algorithm}[!t]
    \caption{Refined Statistics Estimation}
    \label{algorithm:triangle_counting_estimation}
    \KwIn{Users' adjacency lists $\{l_1,l_2,\cdots,l_n\}$, thresholds $\mathbf{\theta}=(\theta_1,\theta_2,\cdots,\theta_n)$, perturbed adjacency matrix $\tilde{M}$, privacy budget $ \varepsilon_3 $, noisy degree $\tilde{d}$}
    \KwOut{Calibrated triangle counts $\bar t$, calibrated degrees $\bar d$}
    \textbf{Step 1: Triangle Count Collection} \\
    \For{$i=1$ to $N$}{
        $\bar{l}_{i} \leftarrow \mathbf{GraphProjection}(l_i,\theta_i)$ \\
        \label{line:graph_projection}
        $t_i\leftarrow \sum_{j=1}^{n} \sum_{k=1}^{j-1}\bar{l}_{i,j} \cdot \bar{l}_{i,k} \cdot \tilde{M}_{j,k}$\\
        $s_i\leftarrow \sum_{j=1}^{n} \sum_{k=1}^{j-1} \bar{l}_{i,j} \cdot \bar{l}_{i,k}$ \\
        $t_i^{\prime} \leftarrow t_i+Lap(\frac{\theta_i}{\varepsilon_3})$, 
        $s_i^{\prime} \leftarrow s_i + Lap(\frac{\theta_i}{\varepsilon_3})$ \\
        
        Calculate $\tilde{t}_i$ based on $t_i^{\prime}$, $s_i^{\prime}$ and \autoref{eq:triangle_calibration} \\
    }
    \textbf{Step 2:  Consistency Restoration} \\
    \label{line:consistentor_begin}
    \For{$i=1$ to $N$}{
        // Find a valid pair according to~\autoref{eq:consistency_restoration} \\
        $(\bar{d_i},\bar{t_i})\leftarrow \mathbf{Consistenter}(\tilde{d}_i, \tilde{t}_i)$ \\
    }
    \label{line:consistentor_end}
          
\end{algorithm}

\subsection{Global Graph Reconstruction} 
\label{section:global_graph_reconstruction}
The final phase focuses on graph synthesis using the neighborhood structure information (degrees $\bar{d}$ and triangle counts $\bar{t}$) obtained in previous phases.
Directly satisfying each node's degree and triangle count simultaneously is intractable.
To address this, we adopt a \textit{triangle-first reconstruction} strategy.
The key insight behind this is that satisfying triangle counts is a more complex and constrained task than satisfying degrees.
If we satisfy degrees first, modifying the graph later to match the target triangles is difficult, as any edge addition can have cascading effects on the degrees and the local clustering structures.
Conversely, by constructing triangles initially, we establish the high-order structural backbone and leave the remaining task of generating necessary edges to fulfill the degree target as a more manageable problem.
The idea is conceptually related to block-based generative models~\cite{karrer2011stochastic, seshadhri2012community, edunov2018generating} where nodes are partitioned into blocks and edges are then generated to match specific structural properties.

The general reconstruction is organized into three steps: node grouping determination, intra-group construction, and inter-group construction.

\begin{algorithm}[!t]
    \caption{Node Grouping Determination}
    \label{algorithm:node_grouping_determination}
    \KwIn{Target degrees $\bar{d}=(\bar{d}_1,\bar{d}_2,\cdots,\bar{d}_n)$, target triangle counts $\bar{t}=(\bar{t}_1,\bar{t}_2,\cdots,\bar{t}_n)$}
    \KwOut{Groups of nodes $S$}
        $S,S^{\prime}\leftarrow\{\}$\\
    \textbf{Step 1: Initial Assignment} \\
    \label{line:assignment_begin}
    \For{$i=1$ to $N$}{
        // Find minimum $d^{\prime}$ satisfying the triangle count \\
        $d' \leftarrow \min \{d \mid \frac{d(d-1)}{2} \ge \bar{t}_i\}$
        
        // Find a non-full group or add a new group to $S$\\
        \If { $\mathrm{ExistGroup(}S^{\prime}, d^{\prime} \mathrm{)}$ is \textbf{false}}{
            $g^{\prime} \leftarrow$ EmptyGroup()\\
            $S^{\prime}$.Add(g)\\
        }
        $g \leftarrow$ SelectGroup($S^{\prime}, d^{\prime}$)\\
        $g$.Append($u_i$)\\
        
        \If {$g.\mathrm{size} > \min_{j\in g}(\bar{d}_j)$}{
            $S^{\prime}$.Remove($g$)\\
            $S$.Add($g$)\\
        }
    }
    \label{line:assignment_end}
    \textbf{Step 2: Incomplete Groups Merging} \\
    Sort the incomplete group $S^{\prime}$ in ascending order of $d'$ \\
    \label{line:merge_begin}
    $g^{\prime}\leftarrow$ EmptyGroup()\\
    \For{$g \in S^{\prime}$}{
        $g^{\prime}$.Merge($g$)\\
        $S^{\prime}$.Remove($g$)\\
        \If {$g^{\prime}.\mathrm{size}$ $>$ $\min_{j\in g^{\prime}}(\bar{d}_j)$}{
            $S$.add($g^{\prime}$)\\
            $g^{\prime} \leftarrow$ EmptyGroup()\\
        }
    }
    \label{line:merge_end}
\end{algorithm}

\mypara{Node Grouping Determination}
To implement the triangle-first strategy, we begin by partitioning nodes to simplify the construction process. 
We group the nodes based on their target triangle count $\bar{t}$ since
this creates sub-problems where all nodes within a group have a similar triangle structure objective.
To achieve a target triangle count of $\bar{t}_i$, a node must have a minimum degree $d^{\prime}$ that satisfies the inequality $\frac{d^{\prime}(d^{\prime}-1)}{2}\geq t$.
Specifically, for a group to be feasible to form the desired triangle structures, each node needs at least $d^{\prime}$ neighbors from the same group.
This establishes a lower bound on the size $k$ of the group: $k\ge d'+1$.

The primary objective of the following intra-group construction phase is to exclusively satisfy the target triangle counts.
Therefore, we need to prevent the degree of any node from violating its target $\bar{d}_i$ during this step.
For instance, if a group has a size of $k$, the intra-group construction will possibly connect each node to the other $k-1$ nodes. 
To ensure no node's degree exceeds its target, the number of newly added intra-group edges $k-1$ must not be greater than the smallest degree budget in the group, which safely tolerates within-group variance by imposing an upper bound on the group size: $k\le \min(\bar{d}_i)+1$.

With these bounds established, the grouping process proceeds in two phases, detailed in~\autoref{algorithm:node_grouping_determination}.
First, we perform an initial assignment by iterating through each node sequentially (\autoref{line:assignment_begin} - \autoref{line:assignment_end}).
For each node $u_i$, we calculate its required intra-group degree $d_i^{\prime}$ as the identical value and place it into the corresponding active group.
Whenever a group's size grows to satisfy its upper bound constraint, it is considered a valid complete group and moved to the final set $S$. 

After all nodes are processed, any remaining active groups that have not met their minimum size requirement are considered illegal and must be merged (\autoref{line:merge_begin} - \autoref{line:merge_end}).
To preserve the homogeneity of triangle targets, all illegal groups are first sorted based on their target triangle counts.
Then we merge the adjacent groups until the new group satisfies the minimum size requirement.

\begin{algorithm}[!t]
        \caption{Intra-group Edge Construction}
        \label{algorithm:intragroup_edge_construction}
        \KwIn{Target triangle counts $\bar{t}=(\bar{t}_1,\bar{t}_2,\cdots,\bar{t}_n)$, the set of groups $S$.}
        \KwOut{Partial synthetic graph $G^{\prime}$}
            \For{$g \in S$}{
                $n\leftarrow g$.size\\
                Calculate $P_e$ based on $n$, $\bar{t}$ and \autoref{eq:intra_edge_p} \\
                \For{$u_i,u_j \in g, i<j$}{
                    \If {$\mathrm{Random}(0,1)<P_e$}{
                        $G^{\prime}$.AddEdge($u_i$,$u_j$)\\
                    }
                }
            }
\end{algorithm}

\begin{algorithm}[!t]
        \caption{Inter-group Edge Construction}
        \label{algorithm:intergroup_edge_construction}
        \KwIn{Target degrees $\bar{d}=(\bar{d}_1,\bar{d}_2,\cdots,\bar{d}_n)$, target triangle counts $\bar{t}=(\bar{t}_1,\bar{t}_2,\cdots,\bar{t}_n)$, the set of groups $S$, partial synthetic graph $G^{\prime}$}
        \KwOut{Final synthetic graph $G$}
            $d_{curr}\leftarrow G^{\prime}$.degrees()\\
            $d_{rem}\leftarrow \bar{d}-d_{curr}$\\
            $G\leftarrow G^{\prime}$\\
            \For{$g \in S$}{
                $F[g]\leftarrow \mathrm{\textbf{FenwickTree}}(g, d_{rem})$\\
                \label{line:fenwicktree}
            }
            
            \For{$g_i,g_j \in S, i<j$}{
                \label{line:intergroup_construction_begin}
                $d_i \leftarrow \sum_{u_k\in g_i}d_{rem}[u_k]$ \\
                $d_j \leftarrow \sum_{u_k\in g_j}d_{rem}[u_k]$ \\
                $E_c \leftarrow \frac{d_i*d_j}{\sum d_{rem}}$\\
                \label{line:candidate_num}
                \For{$k=1$ to $E_{c}$}{                    
                    $u_i \leftarrow F[g_i]$.Select(Random($0$, $d_i$)) \\
                    $u_j \leftarrow F[g_j]$.Select(Random($0$, $d_j$)) \\
                    $G$.AddEdge($u_i$,$u_j$)\\
                    $d_{rem}[i]\leftarrow d_{rem}[i]-1$, $d_{rem}[j]\leftarrow d_{rem}[j]-1$\\
                    $F[g_i]$.Update($u_i$), $F[g_j]$.Update($u_j$)\\
                }
            }
            \label{line:intergroup_construction_end}
\end{algorithm}

\mypara{Intra-group Edge Construction}
With the nodes partitioned into various groups, the second step is to build a dense local structure within each group to primarily satisfy the triangle count targets.
All nodes within a single group share a similar target for their local triangular structure.
This allows us to make a simplified assumption: every node within a group can be treated as equivalent for the purpose of triangle formation.
An Erdős-Rényi (ER) random graph~\cite{erdos1960evolution} $G(n, P_e)$ is the classic model that embodies this principle of statistical uniformity, where every pair of nodes has an independent probability $P_e$ of being connected.
Therefore, we naturally choose the ER model for the step of intra-group edge construction.
To apply this model, we must determine the two required parameters: the number of nodes $n$ and the edge probability $P_e$.
The parameter $n$ is simply the group size.
The $P_e$ is determined by the relationship between the edge probability and the expected number of triangles, as shown in the following formula:
\begin{equation}
\frac{\sum t_i}{3} = P_e^3\frac{n(n-1)(n-2)}{6}.
\label{eq:intra_edge_p}
\end{equation}

As shown in \autoref{algorithm:intragroup_edge_construction}, for each group, we iterate through every pair of nodes within the group and build an edge between them with the calculated probability $P_e$.
This process effectively builds subgraphs where the number of triangles involving each node closely approximates its target $\bar{t}_i$. 

\mypara{Inter-group Edge Construction}
The final step is to build inter-group edges to fulfill the remaining degree requirements for every node, as shown in \autoref{algorithm:intergroup_edge_construction}.
To preserve the triangle counts established in the previous step, edges are generated exclusively between nodes from different groups (\autoref{line:intergroup_construction_begin} - \autoref{line:intergroup_construction_end}), largely avoiding new triangle formation.
In typical real-world networks, most nodes have low degrees, while only a small number of hub nodes possess very high degrees.
A simple random connection mechanism would quickly satisfy the degree requirements of numerous low-degree nodes, leaving the high-degree nodes with insufficient candidates to meet their large targets. 
To ensure that nodes with higher degree deficits are prioritized, we employ a weighted random sampling approach, where a node's probability of being selected is proportional to its current degree deficit.
We implement the weighted sampling using a Fenwick tree~\cite{fenwick1994new}, which maintains the cumulative degree deficits of nodes in each group (corresponding to \textbf{FenwickTree} in \autoref{line:fenwicktree}).
It supports both node selection and dynamic updates in logarithmic time.

In practice, the algorithm iterates through each pair of groups.
For each pair, we first estimate the number of edges $E_c$ to form between them, which is calculated to be proportional to the groups' total remaining degrees (\autoref{line:candidate_num}).
Then we repeatedly sample a node from each group using their respective Fenwick trees, add an edge, and update the deficits and tree structures accordingly until all degree targets $\bar{d}$ are approximately met.
This yields the final synthetic graph.

\subsection{Algorithm Analysis}
\mypara{Privacy Budget Analysis}
Recalling that \method consists of three phases, \ie, initial statistics collection, refined statistics estimation, and global graph reconstruction. 
The first two phases both involve touching users' real data, where each user reports degree, adjacency list, and triangle counts with $\varepsilon_1$, $\varepsilon_2$, and $\varepsilon_3$, respectively.
By the sequential composition property of LDP, these three releases together satisfy $(\varepsilon_1+\varepsilon_2+\varepsilon_3)$-edge LDP.
The final global graph reconstruction is performed by the curator using only the perturbed data, thus not consuming any privacy budget in accordance with the property of post-processing.
Thus, we can obtain~\autoref{throrem:basic_statistics_collection}.
The detailed proof is provided in Appendix~\ref{section:privacy_budget_analysis}.
\begin{theorem}
\label{throrem:basic_statistics_collection}
\method satisfies $(\varepsilon_1+\varepsilon_2+\varepsilon_3)$-edge LDP.
\end{theorem}

\mypara{Complexity Analysis}
We provide an analysis of the computational complexity and communication workload of \method.
Let $N$ be the total number of users and $d$ be the average degree of users.
The complexity of \method is $O(d^2N)$ on the user side and $O(dNlogN)$ on the curator side.
The communication workload is $O(dN)$.
Detailed analyses are provided in Appendices~\ref{section:complexity_analysis} and~\ref{section:communication_overhead}.

\section{Evaluation}
\label{sec:evaluation}

\subsection{Experimental Setup}
\label{subsec:experimental_setup}

\mypara{Datasets}
We run experiments on four real-world datasets. 
\autoref{table:dataset_statistics} shows the basic information of the datasets. 

\mypara{Metrics}
Following~\cite{qin2017ldpgen,nguyen2015differentially,yuan2023privgraph}, we evaluate the quality of the synthesized graphs from multiple perspectives to ensure a comprehensive assessment. 
Each metric is used to measure the similarity between the synthesized graph and the original graph in a specific structural
feature.
We grouped them into three categories.
\begin{itemize}[itemsep=2pt,topsep=2pt,parsep=0pt]
    \item \textbf{Node-level Information.} 
    This category evaluates the properties of individual nodes using two key metrics: degree distribution and eigenvector centrality (EVC).
    For the former, We compute the Kullback-Leibler (KL) Divergence~\cite{kullback1997information}, where a lower value signifies a more faithful reconstruction of the node degree patterns.
    For EVC, which quantifies node influence,
    we report the Mean Absolute Error (MAE) of the top 1\% most influential nodes' EVC scores.

    \item \textbf{Local Clustering Structure.}
    This category quantifies the tendency of nodes to form tight groups, using two related metrics: the total number of triangles and the local clustering coefficient.
    We utilize Relative Error (RE) to evaluate the accuracy of the total triangle count.
    To provide a more detailed measurement, we adopt the local clustering coefficient~\cite{wasserman1994social}, which reflects how densely connected a node's neighbors are to one another.
    In particular, we calculate the Mean Squared Error (MSE) of all nodes' clustering coefficients compared to the original graph.

    \item \textbf{Global Topology Structure.} 
    This category describes macroscopic properties of the graph through three metrics: density, transitivity, and modularity.
    We evaluate density to measure the graph's sparsity.
    To assess the graph's overall tendency to form clusters, we measure transitivity, which calculates the ratio of closed triplets to all connected triplets~\cite{schank2005approximating}.
    Besides, we utilize modularity to quantify the strength of the graph's community structure, with the Louvain method~\cite{blondel2008fast} used for community detection.
    For all three of these global metrics, we assess the performance by computing their RE against the values from the original graph.
\end{itemize}    
For a more comprehensive evaluation of the graph topology preservation, we conduct additional experiments on features such as motifs and diameter, with results presented in Appendix~\ref{section:expanded_metrics}.

\begin{table}[!t]
    \caption{Dataset statistics.}
    \centering
    \small
    \setlength{\tabcolsep}{0.45em}
    \begin{tabular}{c | c | c | c | c}
    \toprule
     \textbf{Dataset} & \textbf{Nodes} & \textbf{Edges} & \textbf{Density} & \textbf{Type} \\
     \midrule
       Chamelon~\cite{rozemberczki2021multi}  & 2,277 & 31,421 & 0.01213 
       & Web page
       \\
       Facebook~\cite{leskovec2012learning}  & 4,039 & 88,234 & 0.01082 
       & Social 
       \\
       CA-HepPh~\cite{leskovec2007graph} & 12,008 & 118,521 & 0.00164 & Collaboration  \\

       Twitter~\cite{leskovec2012learning} & 81,306 & 1,768,149 & 0.00041 & Social  \\ 
      \bottomrule
    \end{tabular}
    
    \label{table:dataset_statistics}
\end{table}

\mypara{Competitors} 
We compare \method with three typical baseline methods (\ie, \dbg, \albg, and \gen\cite{qin2017ldpgen}), which have been introduced in~\autoref{section:existing_solutions}. 

\mypara{Implementation}
For the allocation of the total privacy budget $\varepsilon$, we set $\varepsilon_1:\varepsilon_2 :\varepsilon_3=0.2\varepsilon:0.4\varepsilon:0.4\varepsilon$, 
where $\varepsilon$ ranges from 0.5 to 3.0. 
We implement \method with Python 3.11, and all experiments are conducted on a server with Intel Xeon E5-2620 v4@2.10GHz and 128GB memory.
We repeat the experiment 10 times for each setting, with the mean and standard deviation reported.

\subsection{End-to-End Evaluation}
\label{subsec:end_to_end}
In this section, we perform an end-to-end evaluation of \method and the competitors from three perspectives. 
\autoref{fig:end_to_end} illustrates the experimental results on various datasets.

\begin{figure*}[!t]
\centering
\includegraphics[width=0.95\textwidth]
{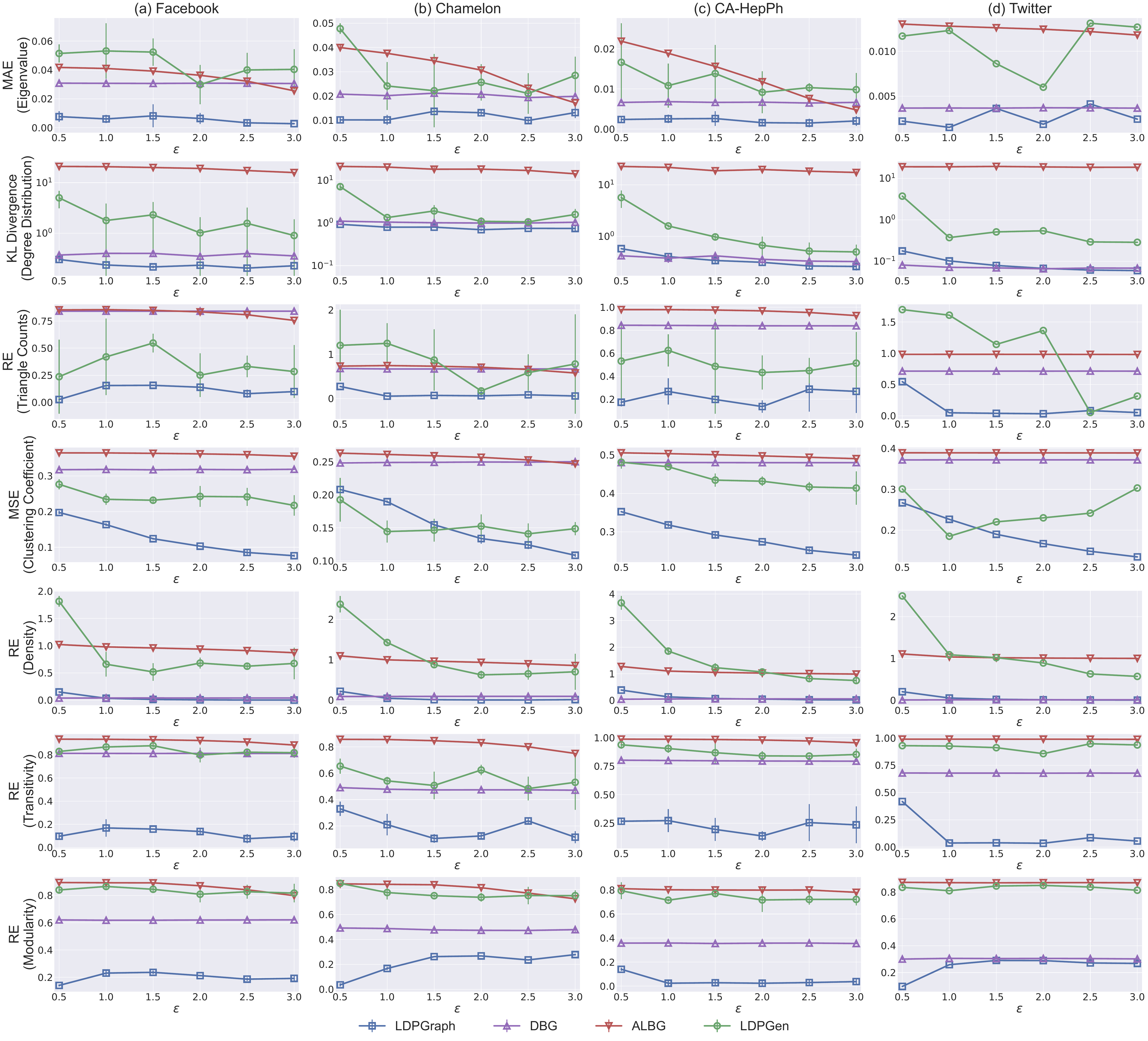}
\caption{
{End-to-end comparison. %
The columns represent the used datasets, and the rows stand for different metrics.
In each plot, the x-axis denotes the privacy budget $\varepsilon$, and the y-axis denotes the performance. For all rows, lower is better.
}}
\label{fig:end_to_end}
\end{figure*}

{
\mypara{Results}
For node-level information, \method demonstrates the highest accuracy across most settings.
As shown in~\autoref{fig:end_to_end}, \method achieves the lowest error for both the MAE of Eigenvector Centrality and the KL Divergence of the Degree Distribution on the Facebook, Chamelon, and CA-HepPh datasets.
An exception is on the Twitter dataset, where \dbg performs slightly better when the privacy budget is small. 
This can be attributed to the fact that \dbg is specifically designed to focus its entire privacy budget and generation model on node degree.
However, \method remains highly competitive in this scenario by capturing richer structural statistics, and even surpasses \dbg on all other datasets.
In contrast, \albg consistently performs the worst, as its bit-flipping mechanism on sparse adjacency lists introduces excessive noise.

The superiority of \method is evident in its preservation of local clustering structures.
Across most settings, \method achieves lower RE of Triangle Count and MSE of Clustering Coefficient than all baseline methods.
This highlights the benefit of collecting triangle-count information.
While \gen and \dbg attempt to reconstruct local cohesion from simpler degree information, they fail to capture the dense neighborhood structures. 

In the domain of global topology structure, \method still demonstrates superior performance. 
For the Density metric, \method accurately matches the original graph's sparsity in all test scenarios.
More importantly, \method establishes an obvious lead on the Transitivity and Modularity. 
On the Facebook dataset with strong community structures, \method significantly reduces the error compared to all baselines.
This result indicates that the graphs synthesized by \method possess far more realistic topological structures than those produced by others.

\mypara{Takeaways} Based on the experimental results, we obtain the following conclusions. 
\begin{itemize}[itemsep=0pt,topsep=2pt,parsep=0pt]
    \item The success of \method highlights a key advantage of using a series of statistics.
    Rather than \albg that protects the entire adjacency matrix with a large amount of noise,
    we protect a compact set of highly informative statistics representing the neighborhood structure, which allows us to concentrate the privacy budget on preserving these properties with less noise.
    This explains the superior performance of \method in most settings.
    
    \item In the final reconstruction phase, %
    \method adopts weighted random sampling for edge generation, which assigns a higher priority to nodes with larger target degrees. 
    In this way, \method avoids scenarios where the high-degree nodes are left with unfulfilled degree requirements in the later stage of synthesis.
    This strategy is directly responsible for the great performance in our experiments, leading to a more accurate degree distribution and overall density.
    
    \item By capturing neighborhood structure information, especially triangle counts, \method preserves the local clustering structure which is an important property of real-world graphs.
    Thus, the superior performance on the Clustering Coefficient and Transitivity metrics is an expected result.
    In contrast, the baseline methods (\ie, \gen and \dbg) that primarily or solely rely on degree information fall short.
    
    \item \method clusters nodes with similar statistical profiles through the node grouping mechanism.
    The edge generation algorithm then naturally forms denser connections within these groups.
    This significantly helps the reconstruction algorithm in identifying community structures and recovering the node information.
    
\end{itemize}
}

\subsection{Ablation Studies}
\label{subsec:ablation_study}

\begin{figure*}[!t]
    \centering
    \includegraphics[width=0.95\textwidth] {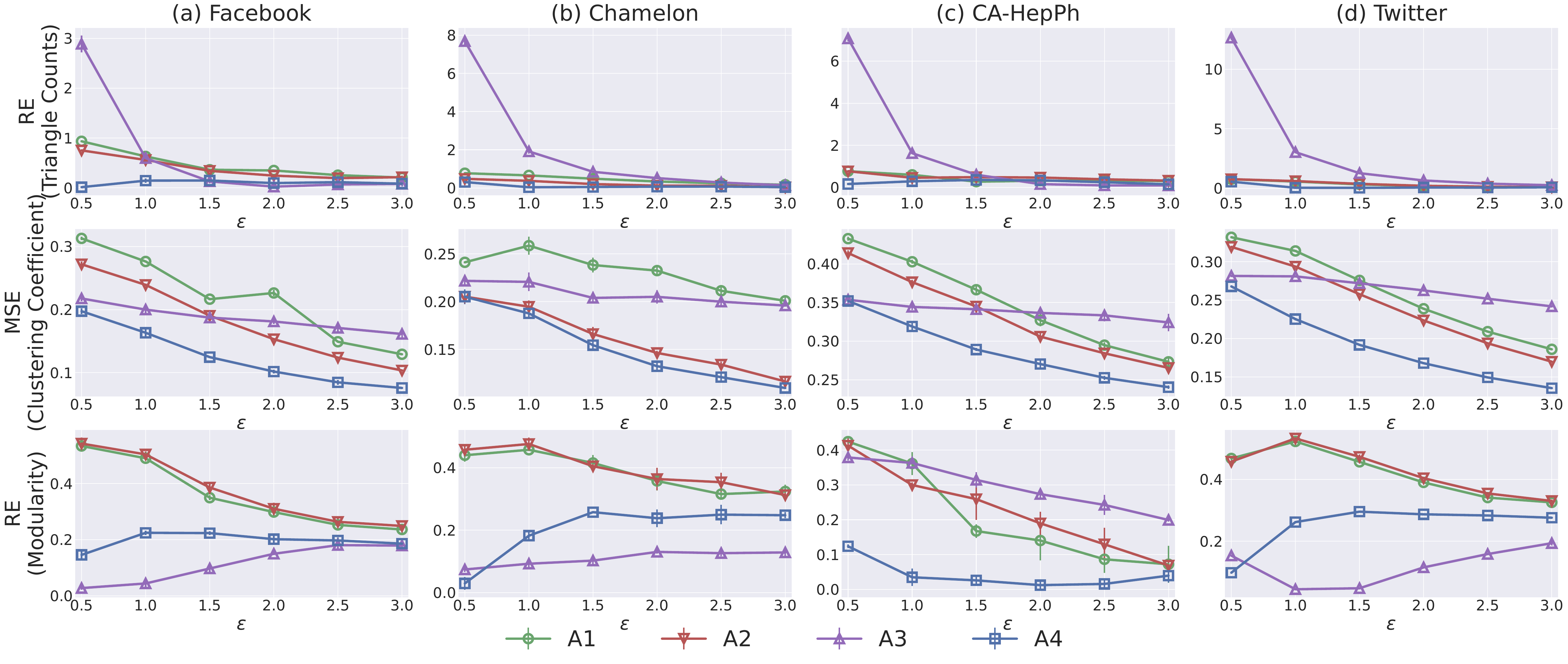}
    \caption{
    Effectiveness of various components.
    For all rows, lower is better.
    }
\label{fig:ablation_study}
\end{figure*}

There are two main components in the first and second phases: The threshold optimization step in Phase 1 and the consistency restoration step in Phase 2.
As detailed in~\autoref{table:ablation_studies}, we design four distinct configurations by selectively enabling or disabling these two modules.

\autoref{fig:ablation_study} illustrates the performance of these variants across three key metrics: Triangle Counts, Clustering Coefficient, and Modularity. 
From the results, our complete solution \textbf{A4}, which integrates both components, achieves the best performance in most cases.
This demonstrates that both steps are essential for accurate graph reconstruction.
Specifically, \textbf{A2} outperforms \textbf{A3} in Clustering Coefficient, suggesting that threshold optimization is more effective for capturing local properties.
Conversely, \textbf{A3} performs better in Modularity and even surpasses \textbf{A4} on several datasets.
This indicates that consistency restoration is crucial for preserving global structures, as it ensures topological feasibility under LDP noise. 
However, \textbf{A3} exhibits poor performance in Triangle Counts at low privacy budgets. 
Considering these diverse strengths and limitations, their integration in \textbf{A4} provides the most robust and balanced performance overall.

\begin{table}[!t]
    \caption{Details of ablation studies.}
    \centering
    \small
    \setlength{\tabcolsep}{0.45em}
    \begin{tabular}{c | c | c }
    \toprule
     \textbf{Name} & \textbf{Threshold Optimization} & \textbf{Structure Consistenter}  \\
     \midrule
       A1 & $\times$ & $\times$ \\
       A2 & \checkmark & $\times$ \\
       A3 & $\times$ & \checkmark   \\
       A4 & \checkmark & \checkmark  \\ 
      \bottomrule
    \end{tabular}
    
    \label{table:ablation_studies}
\end{table}

\subsection{Scalability}
\label{subsec:scalability}
To evaluate the efficiency of \method across different graph scales and densities, we conduct scalability experiments using two sampling strategies on the Facebook and CA-HepPh datasets. 
First, we vary the graph size $N$ by sampling 20\% to 100\% of the nodes while maintaining a constant average degree 
$d$. 
Second, we analyze the graph density $d$ by fixing the node set and scaling the number of edges from 20\% to 100\%. 
The runtimes are reported in \autoref{fig:scalability}.

As illustrated, \method exhibits a linear runtime growth with respect to 
$N$. 
While the runtime increases with d due to the local triangle estimation phase, \method remains computationally feasible for most real-world scenarios where graphs are inherently sparse ($d\ll N$). 
Although its execution time is comparable to ALBG and slower than simpler methods like DBG or LDPGen, \method strikes a superior balance between computational feasibility and structural fidelity.

\begin{figure}[h]
    \centering
    \includegraphics[width=0.45\textwidth] {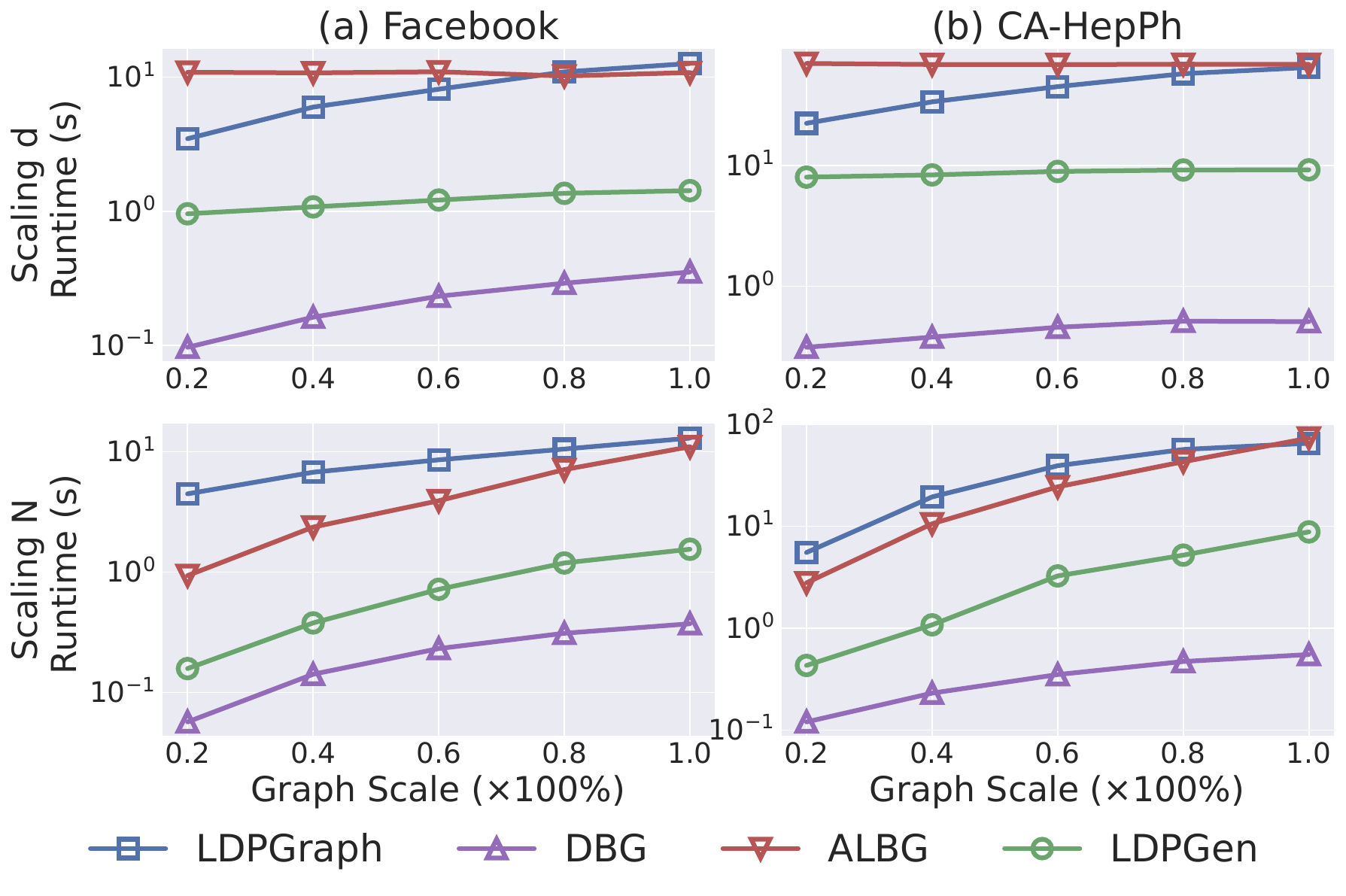}
    \caption{
    Scalability evaluation.
    }
\label{fig:scalability}
\end{figure}

\subsection{Impact of Privacy Budget Allocation}
\label{subsec:impact_of_privacy_budget_allocation}

As detailed in \autoref{sec:methodology}, the total privacy budget $\varepsilon$ is partitioned into three key components: $\varepsilon_1$ for the degree collection, $\varepsilon_2$ for the adjacency list collection, and $\varepsilon_3$ for the triangle count estimation. 
The allocation of this budget is critical as it directly influences the accuracy of the raw statistics used for graph reconstruction.
We fix the total budget $\varepsilon=2.0$ and vary the ratios of $\varepsilon_1/\varepsilon$ and $\varepsilon_2/\varepsilon$ from 0.1 to 0.8 in increments of 0.1, with $\varepsilon_3$ being the remainder.
\autoref{fig:impact_eps} shows the impact across our four datasets.

We observe that assigning a very small $\varepsilon_1$ severely degrades the performance of degree distribution.
We also observe that when the combined budget for degrees and adjacency lists ($\varepsilon_1+\varepsilon_2$) is large, the remaining budget for triangle estimation $\varepsilon_3$ becomes insufficient.
This leads to performance degradation in structural metrics that rely on neighborhood structure, like clustering coefficient and modularity.
These trends indicate that while all components are necessary, the more complex structural statistics are more sensitive to perturbation noise.

Based on these observations, we determined that an unbalanced allocation favoring the more sensitive structural components yields the best overall performance.
Thus, we set the allocation ratio as $\varepsilon_1:\varepsilon_2 :\varepsilon_3=0.2\varepsilon:0.4\varepsilon:0.4\varepsilon$.
This strategy allocates a sufficient budget to more fine-grained statistics of adjacency list and more complex task of triangle count estimation.
The complete experimental results for all evaluated metrics under different budget distributions are provided in Appendix~\ref{section:budget_allocation_results}.

\begin{figure*}[!t]
    \centering
    \includegraphics[width=0.96\textwidth] {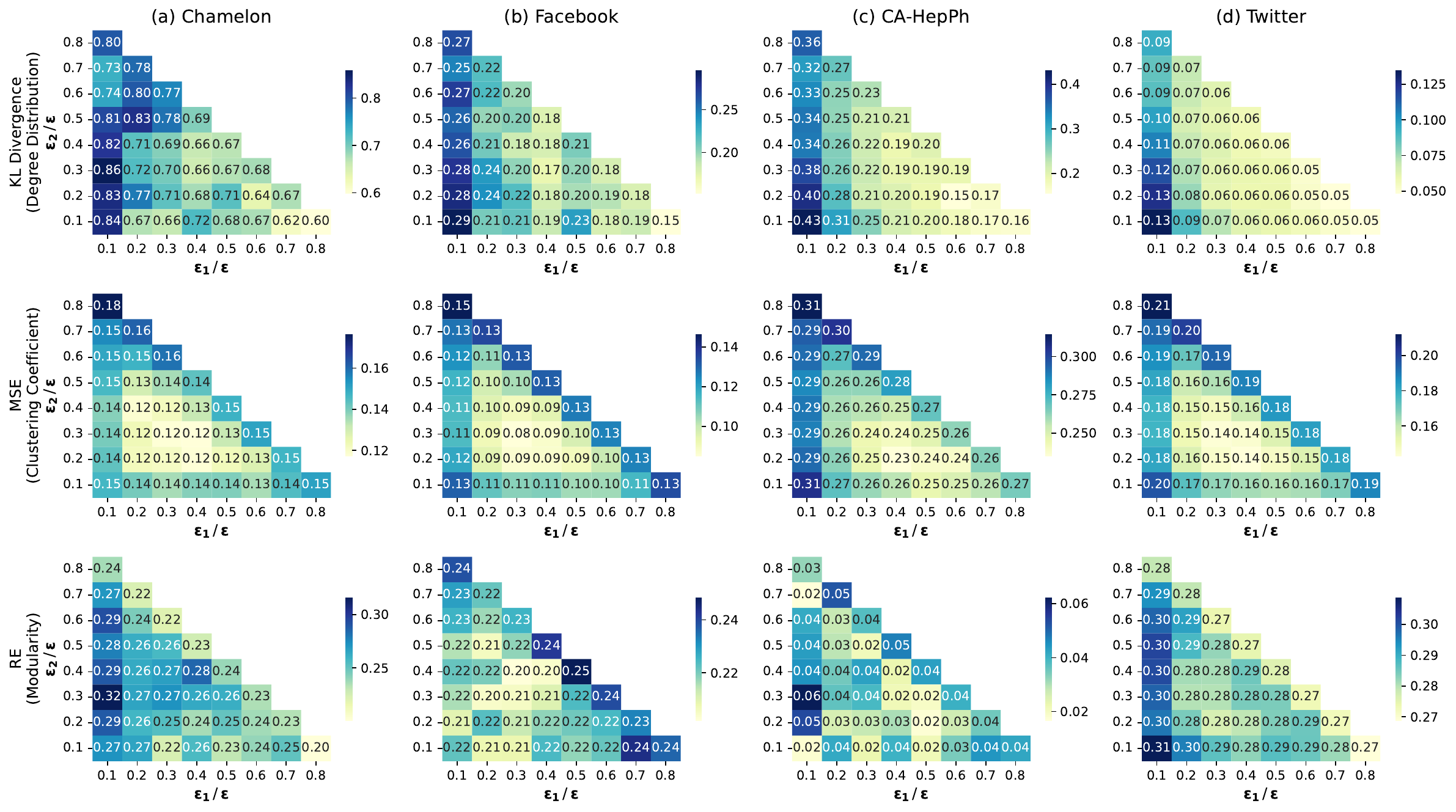}
    \caption{
    Impact of privacy budget allocation. 
    For all grids, lower is better.
    }
\label{fig:impact_eps}
\end{figure*}

\section{Related Work}
\label{sec:related_work}
\subsection{DP Graph Synthesis}
Research in differentially private graph synthesis has largely focused on the centralized privacy model~\cite{nguyen2015differentially, chen2014correlated, xiao2014differentially, clauset2008hierarchical, sala2011pygmalion, mir2012differentially, wang2013preserving, yuan2023privgraph,ji2021xor,yuan2024psgraph}.
One common strategy directly perturbs the adjacency matrix.
For example, the \textit{Top-m Filter} method by Nguyen \etal~\cite{nguyen2015differentially} directly perturbs the adjacency matrix by adding Laplace noise to each cell and constructs the graph from the $m$ potential edges with the highest scores.
Chen \etal~\cite{chen2014correlated} proposed the \textit{density-based exploration and reconstruction} (DER) algorithm. 
DER first reorders nodes to cluster edges, then employs a noisy quadtree to identify dense regions in the adjacency matrix, and reconstructs edges within these regions.

To mitigate the impact of noise, other works turn to generative graph models.
Mir \etal~\cite{mir2012differentially} used the Stochastic Kronecker Graph model~\cite{mahdian2007stochastic}, by privately estimating its parameters to form a probability distribution over the graph. 
This method then generates a synthetic graph by sampling from the distribution.
Yuan \etal~\cite{yuan2023privgraph} further proposed PrivGraph.
PrivGraph preserves community structure by privately detecting communities and using intra- and inter-community connections to reconstruct the synthetic graph.

\subsection{LDP Graph Analysis}
There has also been a handful of works~\cite{sun2019analyzing,imola2021locally,imola2022communication,ye2020towards,ye2020lf,wei2020asgldp,brito2023global} that have focused on estimating specific graph statistics for particular analytical tasks in the LDP setting.
Sun \etal~\cite{sun2019analyzing} proposed a method for subgraph counting that relies on an extended local view, assuming that each user can see all the connections among their friends.
Ye \etal~\cite{ye2020lf} designed a one-round framework to derive complex target metrics, like clustering coefficient, from two privately collected atomic metrics: the node degree and the adjacency list.
Imola \etal~\cite{imola2022communication} presented the related approaches for counting subgraphs like triangles by leveraging an additional round of interaction between users and the data curator.

\section{Conclusion}
\label{sec:conclusion}
In this paper, we propose \method to address the critical challenge of synthesizing graph data with high fidelity under the strict constraints of local differential privacy.
\method first privately collects the neighborhood structures of the original graph, which in particular are node degrees and triangle counts, then utilizes a specialized reconstruction algorithm to generate a synthetic graph with this information.  
Comprehensive experiments conducted on four real-world datasets demonstrate the superiority of \method compared to the baselines, reflecting its ability to strike a significantly better balance between privacy protection and data fidelity.
This work paves the way for more practical and reliable privacy-preserving graph data analysis in the local view settings. 

\section*{Acknowledgments}
This work was supported by the National Science and Technology Major Project (No. 2026ZD0128300), 
the National Natural Science Foundation of China (No. U23A20296, 62402431, 62441618), 
Zhejiang Province's ``Lingyan'' R\&D Project (No. 2025C01195), 
the Zhejiang Provincial Natural Science Foundation of China (No. LZ26F020002), 
the Natural Science Foundation of Jiangsu Province (No. BK20220075).

\bibliographystyle{IEEEtran}
\bibliography{reference}

\newpage

\appendices

\renewcommand{\thesection}{\Alph{section}}
\renewcommand{\thesubsection}{\thesection.\arabic{subsection}}
\makeatletter
\def\@seccntformat#1{%
  \csname the#1\endcsname\;  %
}
\makeatother

\section{Theoretical Analysis}
\label{section:theoretical_analysis}
In Phase 1, to dynamically determine an effective threshold $\theta$ for triangle counting, we conduct a theoretical analysis by considering the trade-off between projection error and noise error.
Our objective is to find a suitable $\theta$ that minimizes the combined error of the final triangle count estimation.
While a universal closed-form error bound is challenging to derive under LDP due to structural dependencies, we construct a tractable objective function based on heuristic approximations.
The combined error of an estimator $\tilde{t}$ for the triangle count is decomposed into its projection error and noise error as follows:
\begin{align*}
E_{comb} &= E_{proj}+E_{noise}%
\\
&= (\tilde{t}-t)^2 
+ \mathbb{E}\left[Lap(\frac{\theta}{\varepsilon_3})^2\right],
\end{align*}
where the former term is the projection error that captures the information loss incurred when $\theta<d$, while the latter term is the noise error introduced by the Laplace noise added to the estimation.

We first analyze the expectation of the projection error.
When calculating the threshold at the curator, we can only access the noisy degrees.
Furthermore, there is no deterministic formula that links the degrees to the triangle counts.
We thus use $g(d)=\frac{d(d-1)}{2}$ as an analytical proxy for the true triangle count $t$.
There are two reasons.
Firstly, it provides a simple differentiable function of $d$.
Secondly and most importantly, it represents the theoretical maximum of the triangle counting estimation for a node with degree $d$, providing a safe basis.
Therefore, we have the following:
\begin{equation*}
E_{proj} \approx \left(g(\theta)-g(d)\right)^2.
\end{equation*}

To simplify the derivation, let $\theta$ be $\tilde{d}+\tau$, where $\tau$ is a tunable correction factor.
We have $\theta = d+Lap(\frac{1}{\varepsilon_1})+\tau$.
Due to the difficulty of directly integrating $(g(\theta)-g(d))^2$, we use a first-order Taylor approximation of $g(\theta)$ around the point $d$.
Let \(X \sim Lap\left(\frac{1}{\varepsilon_1}\right)\). Then
\(\theta=\tilde{d}+\tau=d+X+\tau\). Since \(g'(d)=d-\frac{1}{2}\), we have
\[
g(\theta)-g(d)
\approx g'(d)(\theta-d)
=
\left(d-\frac{1}{2}\right)(X+\tau).
\]
To find \(E_{proj}\), we compute the expectation conditioned on
\(\theta<d\), i.e., \(X<-\tau\), as the projection error is zero when $\theta \geq d$:
\begin{align*}
E_{proj}
&\approx \left(d-\frac{1}{2}\right)^2
\mathbb{E}\left[(X+\tau)^2 \cdot \mathbf{1}\{X<-\tau\}\right] \\
&= \left(d-\frac{1}{2}\right)^2
\int_{-\infty}^{-\tau} (z+\tau)^2 f_X(z)\mathrm{d}z \\
&= \left(d-\frac{1}{2}\right)^2
\int_{-\infty}^{-\tau} (z+\tau)^2
\frac{\varepsilon_1}{2}e^{\varepsilon_1 z}\mathrm{d}z \\
&= \left(d-\frac{1}{2}\right)^2
\frac{1}{\varepsilon_1^2}e^{-\tau\varepsilon_1}.
\end{align*}

where $\mathbf{1}(\cdot)$ is an indicator function. 
This final term shows that the projection error decreases exponentially as $\tau$ increases. 
A larger $\tau$ raises the projection threshold, making it less likely that the local ego-graph will be pruned, thus reducing the information loss.

The noise error term is the expected squared value of the Laplace noise added to the triangle counting estimation. 
We have the following:
\begin{align*}
E_{noise} &= \mathbb{E}\left[ Lap\left(\frac{\theta}{\varepsilon_3}\right)^2 \right] \\
&= \mathbb{E}\left[ 2 \left( \frac{\theta }{\varepsilon_3} \right)^2 \right] \\
&= \mathbb{E}\left[ 2 \left( \frac{	d+X+\tau }{\varepsilon_3} \right)^2 \right] \\
&= \frac{2}{\varepsilon_3^2} \mathbb{E}\left[ \left(d+X+\tau \right)^2\right] \\
&= \frac{2}{\varepsilon_3^2} \left(d^2 + 2d\tau + \tau^2 + \frac{2}{\varepsilon_1^2}\right),
\end{align*}
This formula shows that the noise error increases as $\tau$ increases, because a larger $\tau$ introduces a higher sensitivity for the triangle counting, thus increasing the noise.

By combining the terms of projection error and noise error, we construct the final error function:
\begin{equation*}
f(\tau)=(d - \frac{1}{2})^2 \frac{1}{\varepsilon_1^2} e^{-\tau \varepsilon_1}+\frac{2}{\varepsilon_3^2} (d^2 + 2d\tau + \tau^2 + \frac{2}{\varepsilon_1^2}).
\end{equation*}
Under the LDP setting, the curator only has access to the noisy degree $\tilde{d}$, thus we use $\tilde{d}$ as a practical proxy of $d$ to calculate the threshold.
This function can be minimized with respect to $\tau$ to balance the two types of error, producing an effective projection threshold.

\section{Algorithm Analysis}
\subsection{Privacy Budget Analysis}
\label{section:privacy_budget_analysis}
Recalling that \method consists of three phases, \ie, initial statistics collection, refined statistics estimation, and global graph reconstruction. 
The first two phases can jointly be referred to as the data collection stage, as they both contact with users' real data.
In data collection stage, each user perturbs and reports degree, adjacency list and triangle counts with $\varepsilon_1$, $\varepsilon_2$ and $\varepsilon_3$ respectively.
The final global graph reconstruction is performed by the curator using only the perturbed data, thus not consuming any privacy budget in accordance with the property of post-processing.
Thus, We have the following theorem.
\setcounter{theorem}{0}
\begin{theorem}
\label{theorem:basic_statistics_collection}
\method satisfies $(\varepsilon_1+\varepsilon_2+\varepsilon_3)$-edge LDP.
\end{theorem}

\begin{proof}
\label{proof:basic_statistics_collection}
In the process of initial statistics collection, \method will perturb the actual information with different mechanisms.
The perturbation of degrees is achieved by adding Laplace noise $Lap(\frac{1}{\varepsilon_1})$, where the sensitivity is $1$.
The perturbation of adjacency list is achieved by performing Randomized Response, which is in formula: 
\begin{equation*}
\tilde{l}_i = (RR_{\varepsilon_2}(l_{i,1}),RR_{\varepsilon_2}(l_{i,2}),\cdots,RR_{\varepsilon_2}(l_{i,i-1})).
\end{equation*}

Recalling \autoref{subsec:local differential privacy}, Laplace Mechanism and Randomized Response can provide rigorous local differential privacy guarantee.
Hence, the initial statistics collection process satisfies $(\varepsilon_1+\varepsilon_2)$-edge LDP according to the property of sequential composition.

In the phase of refined statistics estimation, the perturbation of triangle counts is achieved by adding Laplace noise $Lap(\frac{\theta}{\varepsilon_3})$.
The core challenge in protecting the triangle count is the high sensitivity.
A single change to the connections between a user and its neighbors can affect the number of triangles in which it participates by a maximum of the degree of the user.
We set a projection threshold for each node to ensure that the sensitivity will not be violated, \ie, $\theta\geq d'$, where $d'$ represents the degree of the node after projection in the local ego-graph.
Hence, the refined statistics estimation process satisfies $\varepsilon_3$-edge LDP.

According to the property of sequential composition, \method satisfies $\varepsilon$-edge LDP, where the total privacy budget $\varepsilon=\varepsilon_1+\varepsilon_2+\varepsilon_3$.
\end{proof}

\subsection{Complexity Analysis}
\label{section:complexity_analysis}
We provide an analysis of the computational complexity of \method, breaking it down into two primary parts: the workload on the user side and the workload on the curator side.
Let $N$ be the total number of nodes and $M$ be the total number of edges in the original graph.

The computations performed by each user are considered parallel across all users.
On the user side, the process begins with collecting and reporting local statistics.
In the phase of initial statistics collection, each user first counts and perturbs their own degree $d_i$, which takes $O(N)$ time, and then applies Randomized Response to their adjacency list, also taking $O(N)$ time.
In the phase of refined statistics estimation, after receiving feedback from the curator, each user projects their own local ego-graph to ensure that the degree does not exceed the threshold $\theta$, which takes $O(N)$ time.
$2$-star count and triangle count require checking for edges in the perturbed adjacency matrix between all pairs of the user's neighbors, which involves $\frac{d^{\prime}*(d^{\prime}-1)}{2}$ look-ups.
$d^{\prime}$ is upper-bounded by $\theta$.
The worst-case complexity across all users is determined by the user with the largest pruned degree, which could be $O(N^2)$.
In fact, $\theta$ is typically much smaller than $N$, making it highly efficient for most users.
Adding up the complexity of all user-side tasks, we have $O(N^2)$.
For the vast majority of nodes in sparse real-world networks, their low degree keeps this process computationally feasible on the user side.

On the curator side, the complexity consists of the following steps.
Initially, the curator receives noisy degrees and adjacency lists from all users.
It then calculates a threshold for each user, which takes $O(N)$ time, and aggregates the adjacency lists into a noisy adjacency matrix, which also takes $O(N)$ time.
After receiving the noisy triangle counts, the curator performs consistency restoration, which requires a constant-time check and calibration over each user, resulting in a complexity of $O(N)$.

The curator then proceeds to the main graph reconstruction phase.
This involves iterating through each node and assigning it to a group, which is an $O(N)$ operation.
During the intra-group edge construction, the curator iterates through each group and considers building edges between every pair of nodes within the group.
In the worst case, where all nodes are assigned to a single group, the complexity is $O(N^2)$.
In the stage of inter-group edge construction, the curator adds edges between nodes in different groups to satisfy the remaining degree requirement. 
To ensure that nodes with higher degree deficits are prioritized, we employ a weighted random sampling method based on the Fenwick tree. 
Building the structure for each group takes $O(NlogN)$.
At most, we need to sample $M$ pairs of nodes, assuming that we did not build any edge in the stage of intra-group edge construction, which takes $O(MlogN)$.
Considering that most real-world graphs are sparse, it is subsumed by the $O(N^2)$ of the previous step.
Adding up the complexity of all curator side tasks, we have $O(N^2)$.
Actually, the node grouping strategy based on heterogeneous triangle counts effectively prevents the worst-case scenario of a single large group, leading to a much more efficient reconstruction process in practice.

\subsection{Communication Overhead}
\label{section:communication_overhead}
\method's communication overhead consists of two rounds under the LDP framework. 
In the first round, each user uploads their perturbed degree and adjacency list, incurring an $O(N)$ cost. 
The second round introduces the primary bottleneck: 
the curator broadcasts the optimized projection threshold and the aggregated adjacency matrix.
Rather than broadcast a full $O(N^2)$ matrix, we utilize sparse encoding to transmit only the non-zero entries of the aggregated edges.
This ensures that the communication overhead is proportional to the density of the graph, reducing the overhead to $O(dN)$, making \method scalable to large-scale sparse graph.
Finally, each user uploads their single perturbed triangle count at an $O(1)$ cost. 

\section{Extended Experimental Results}
\subsection{Evaluation on Expand Metrics}
\label{section:expanded_metrics}
To provide a more comprehensive assessment of the structural fidelity of the synthetic graphs, we extend our evaluation to include broader graph properties: graph diameter and motif frequency.
Graph diameter serves as a representative metric for global connectivity, while motif frequency captures complex, higher-order local structural patterns.

\begin{itemize}[itemsep=2pt,topsep=2pt,parsep=0pt]
    \item \textbf{Graph Diameter.} 
    We measure the Relative Error (RE) between the diameter of the original graph and the synthesized graph. 
    A lower RE indicates better preservation of macroscopic connectivity and global reachability.
    \item \textbf{Motif Frequency Similarity.} 
    We extract all 3-node and 4-node motifs and compute the Cosine Similarity between the normalized motif frequency vectors of the original and synthetic graphs. 
    A higher similarity implies that the complex, high-order local subgraph structures are well-preserved.
\end{itemize}    

\autoref{fig:expand_metrics} presents the experimental results on the Facebook, Chamelon, and CA-HepPh datasets across varying privacy budgets $\varepsilon\in[0.5, 3.0]$. 
As shown in the top row of \autoref{fig:expand_metrics}, \method consistently outperforms all baselines in terms of Diameter across all evaluated datasets. 
Notably, \dbg exhibits significantly high errors and large variance, indicating its inability to capture long-range paths under LDP constraints. 
\albg and \gen perform relatively better but still fall behind our method. 
In contrast, \method maintains a highly stable and low relative error, demonstrating that our generation model effectively reconstructs the global connectivity of the network.
Regarding higher-order local structures, \method achieves the highest Cosine Similarity for 3- and 4-node motifs across all settings, rapidly converging to and maintaining a high level even under small privacy budgets.
In contrast, baselines struggle significantly to maintain motif consistency.

These demonstrates that our coupled consistency restoration module and generation mechanism effectively reconcile the geometric conflicts introduced by LDP noise, thereby successfully preserving both dense local structures and valid global connectivity.

\begin{figure*}[!t]
    \centering
    \includegraphics[width=0.85\textwidth] {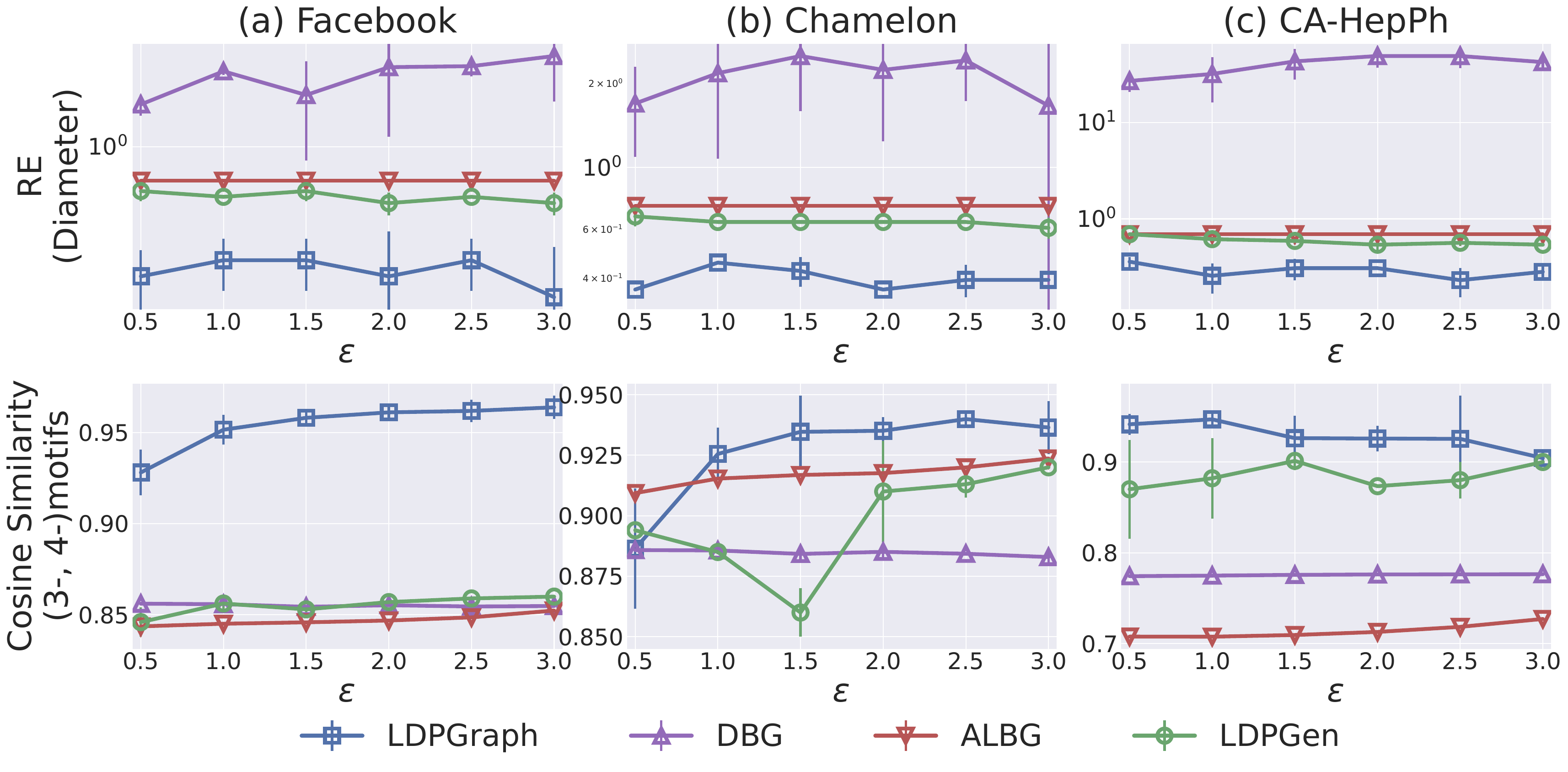}
    \caption{
    Comparison of additional metrics.
    The columns represent the used datasets, and the rows stand for different metrics.
    In each plot, the x-axis denotes the privacy budget $\varepsilon$, and the y-axis denotes the performance.
    For Diameter, lower is better; for (3-,4-) motifs, higher is better.
    }
\label{fig:expand_metrics}
\end{figure*}

\subsection{Results of Privacy Budget Allocation}
\label{section:budget_allocation_results}
\autoref{fig:impact_eps_full} presents a comprehensive view of the impact of privacy budget allocation.
The experimental results reveal that the impact of budget distribution is non-uniform across different graph properties.
Specifically, we observe that certain metrics, such as clustering coefficient and degree distribution, exhibit high sensitivity to the allocation strategy, showing clear performance gradients that necessitate a balanced budget for both local edges and neighborhood structures. 
In contrast, other statistics, such as Eigenvalue Centrality, appear to be relatively robust to the specific budget distribution.
Consequently, our adopted allocation ratio of privacy budget is strategically designed to prioritize the utility of those highly sensitive features while maintaining a competitive and reasonable performance level for the more stable metrics.

\begin{figure*}[!t]
    \centering
    \includegraphics[width=0.94\textwidth] {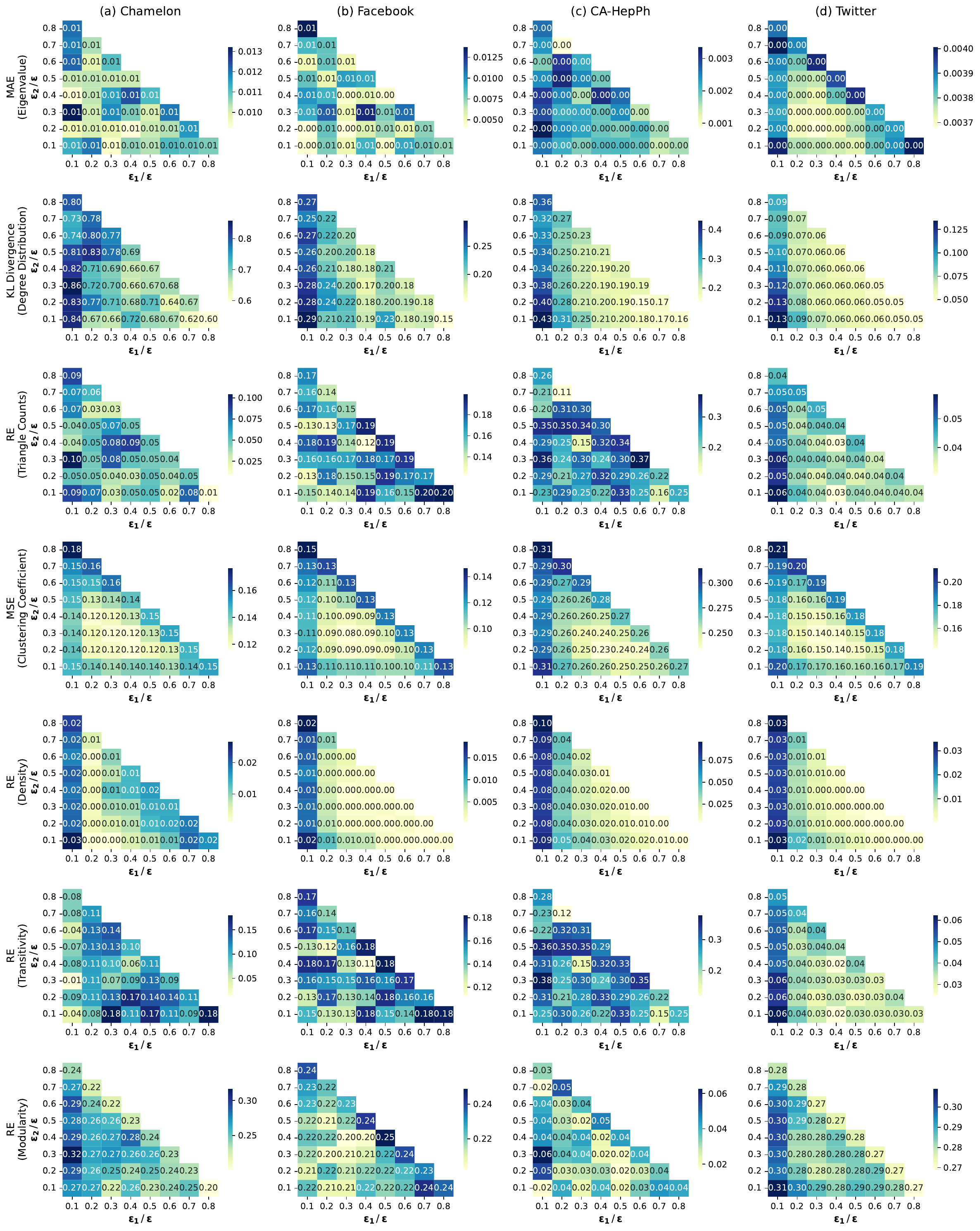}
    \caption{
    Impact of privacy budget allocation. 
    The columns represent the used datasets, and the rows stand for different metrics. 
    In each plot, the x-axis denotes the ratio of $\varepsilon_1$, the y-axis denotes the ratio of $\varepsilon_2$ $(\varepsilon_3=\varepsilon-\varepsilon_1-\varepsilon_2)$, and the values in the grid denote the performance.
    For all grids, lower is better.
    }
\label{fig:impact_eps_full}
\end{figure*}

\end{document}

%% file: mydefs.tex
\usepackage{epsfig,endnotes,xspace,url,diagbox}
\usepackage{amsmath,amsthm,amsfonts}

\usepackage{enumerate}
\usepackage{varioref}
\usepackage{graphicx}

\usepackage{caption}
\usepackage{subcaption}
\usepackage{float}
\usepackage{epstopdf}
\usepackage{xcolor}
\usepackage[most]{tcolorbox}
\usepackage{multirow}
\usepackage{comment}
\usepackage{color}
\usepackage[utf8]{inputenc}
\usepackage{mathtools}
\usepackage{wrapfig}
\usepackage{tabularx}
\usepackage{xspace}
\usepackage[noend, ruled,linesnumbered]{algorithm2e}

\usepackage{hyperref}

\usepackage{cases}
\usepackage{stfloats}
\usepackage{fancyhdr}
\usepackage{booktabs}

\usepackage{enumitem}
\setlist[itemize]{leftmargin=*}

\hypersetup{
  colorlinks,
  linkcolor={blue!70!green},
  citecolor={green!70!blue},
  urlcolor={orange!70!red}
}

\newif\ifcomment
\commenttrue 
\ifcomment
	\newcommand{\says}[2]{\noindent\textcolor{purple}{\textbf{#1 says: }}\textcolor{blue}{#2}\xspace}
	\newcommand{\zzk}[1]{\says{Zhikun}{#1}}
    \newcommand{\mc}[1]{\says{Min}{#1}}	
    \newcommand{\yq}[1]{\says{Quan}{#1}}	
    	
\else  
    \newcommand\says[2]{}
    \newcommand\zzk[1]{}
    \newcommand\mc[1]{}
    \newcommand\yq[1]{}
\fi

\newcommand{\etal}{\textit{et al.}\xspace}
\newcommand{\ie}{\textit{i.e.}\xspace}
\newcommand{\eg}{\textit{e.g.}\xspace}

\newcommand{\etc}{\textit{etc.}\xspace}

\newcommand{\mypara}[1]{\noindent\textbf{#1.} \xspace}

\renewcommand{\Pr}[1]{\ensuremath{\mathsf{Pr}\left[#1\right]}\xspace}

\newtheorem{theorem}{Theorem}

\newtheorem{definition}{Definition}

\definecolor{revision}{RGB}{0,0,0}

\newcommand{\revstart}{\begin{color}{revision}}
\newcommand{\revend}{~\!\!\end{color}}

\newcommand{\method}{\ensuremath{\mathsf{LDPGraph}}\xspace}

\newcommand{\gen}{\ensuremath{\mathsf{LDPGen}}\xspace}

\newcommand{\dbg}{\ensuremath{\mathsf{DBG}}\xspace}
\newcommand{\albg}{\ensuremath{\mathsf{ALBG}}\xspace}

\usepackage{etoolbox}
\makeatletter
\patchcmd{\hyper@makecurrent}{%
    \ifx\Hy@param\Hy@chapterstring
        \let\Hy@param\Hy@chapapp
    \fi
}{%
    \iftoggle{inappendix}{%
        \@checkappendixparam{chapter}%
        \@checkappendixparam{section}%
        \@checkappendixparam{subsection}%
        \@checkappendixparam{subsubsection}%
        \@checkappendixparam{paragraph}%
        \@checkappendixparam{subparagraph}%
    }{}%
}{}{\errmessage{failed to patch}}

\newcommand*{\@checkappendixparam}[1]{%
    \def\@checkappendixparamtmp{#1}%
    \ifx\Hy@param\@checkappendixparamtmp
        \let\Hy@param\Hy@appendixstring
    \fi
}
\makeatletter

\newtoggle{inappendix}
\togglefalse{inappendix}

\apptocmd{\appendix}{\toggletrue{inappendix}}{}{\errmessage{failed to patch}}